\RequirePackage{amsmath}
\documentclass[runningheads]{llncs}

\usepackage[T1]{fontenc}
\usepackage{graphicx}
\usepackage{amsmath}
\usepackage{amsfonts}
\usepackage{amssymb}
\usepackage{fancyhdr}
\usepackage{hyperref}
\usepackage{lipsum}
\usepackage{float}
\usepackage[colorinlistoftodos]{todonotes}

\usepackage{amsthm}
\usepackage[mathscr]{euscript}
\usepackage{subcaption}

\usepackage{apptools}
\AtAppendix{\counterwithin{lemma}{section}}
\AtAppendix{\counterwithin{proposition}{section}}
\AtAppendix{\counterwithin{remark}{section}}
\AtAppendix{\counterwithin{theorem}{section}}

\newcommand{\I}[2]{I_{#1,#2}}
\newcommand{\Dif}[2]{B_{#1}-B_{#2}} 
\newcommand{\Dift}[3]{B_{#1}^{#3}-B_{#2}^{#3}} 
\newcommand{\Fdif}[2]{\beta_{{#2},{#1}}(\Dif{#1}{#2})} 
\newcommand{\Fdift}[3]{\beta_{{#2},{#1}}(\Dift{#1}{#2}{#3})} 
\newcommand{\B}[1]{B_{#1}} 
\newcommand{\Bt}[2]{B_{#1}^{#2}} 

\newcommand{\R}{\mathbb{R}} 
\newcommand{\N}{\mathbb{N}} 

\newcommand{\SR}{\mathbf{S}}
\newcommand{\Mal}{\mathbf{M}}
\newcommand{\Rec}{\mathbf{R}}
\newcommand{\BF}{\mathbf{B}}
\newcommand{\Ins}{\mathbf{I}}

\newcommand{\bias}[1]{\beta_{#1}}

\newcommand{\source}{\mathcal{S}}

\usepackage{svg}
\usepackage{tikz}
\usetikzlibrary{automata,arrows,calc,shapes,decorations.pathreplacing,patterns}
\tikzset{->,>=stealth',auto,node distance=2cm,thick,initial text=}
\tikzstyle{accepting}=[path picture={%
  \draw let 
    \p1 = (path picture bounding box.east),
    \p2 = (path picture bounding box.center)
    in
      (\p2) circle (\x1 - \x2 - 2pt);
  }]
\tikzset{
    position/.style args={#1:#2 from #3}{
        at=(#3.#1), anchor=#1+180, shift=(#1:#2)
    }
}

\definecolor{red1}{RGB}{163, 0, 0}
\definecolor{red2}{RGB}{196,70,1}
\definecolor{red3}{RGB}{245,118,0}
\definecolor{blue1}{RGB}{5,79,185}
\definecolor{blue2}{RGB}{0,115,230}
\definecolor{blue3}{RGB}{139,171,241}

 \def\Bleck{} 

\usepackage[breakable,skins]{tcolorbox} 

\usepackage{soul} 
\usepackage{twoopt} 
\usepackage{varwidth}  
\usepackage{lipsum} 
\usepackage{xspace} 

\def\ReviewsOn{} 

\def\CommentsOn{} 

\ifdefined\ForceBleck 
	\def\Bleck{}
	\undef\CommentsOn{}
	\undef\ReviewsOn{} 
\else 
	\ifdefined\ForceDraft 
		\undef\Bleck{} 
		\def\CommentsOn{}
		\def\ReviewsOn{} 
	\else
		\ifdefined\ForceReview 
			\undef\Bleck{} 
			\undef\CommentsOn{}
			\def\ReviewsOn{} 
		\fi	
	\fi
\fi

\ifdefined\Bleck 
	\newcommand{\DraftOnly}[1]{} %
	\undef\CommentsOn{} 
	\undef\ReviewsOn{} 
\else
	\ifdefined\ForceReview 
		\newcommand{\DraftOnly}[1]{} %
		\def\ReviewsOn{} 
	\else 
		\newcommand{\DraftOnly}[1]{#1}
	\fi
\fi

\ifdefined\ReviewsOn 
	\newcommand{\review}[1]{
		{{\leavevmode\color{magenta}{#1}}}\xspace
	}
	\newcommand{\reviewfootnote}[1]{\footnote{\review{\textbf{Nota de revisão:} #1}}} 
\else
	\newcommand{\review}[1]{#1} 
	\newcommand{\reviewfootnote}[1]{} 
\fi

\ifdefined\CommentsOn 

\else 

\fi 

\ifdefined\CommentsOn 
	\newcommand{\draftbox}[2]{ 
		\begin{center}
			\begin{tcolorbox}[
				enhanced jigsaw,
				breakable,
				left=2pt,
				right=2pt,
				top=2pt,
				bottom=2pt,
				bicolor,
				colframe = black!70,
				colbacklower = black!70,
				colback  = purple!5!white,
				coltitle = white,
				title = {Begin draft \texttt{$\blacktriangleright$#1}},
				halign lower = right,
				]
				\textcolor{black}{\texttt{#2}}
				\tcblower
				\textcolor{white}{\texttt{End draft}}
			\end{tcolorbox}					
		\end{center} 
	}
\else 
	\newcommand{\draftbox}[1]{} 
\fi 

\newtcolorbox{todobox}[3][]
{
	enhanced jigsaw,
	enforce breakable,
	hbox,
	left=1pt,
	right=1pt,
	top=1pt,
	bottom=1pt,
	colframe = #2!20,
	colback  = #2!100,
	coltitle = #2!20!black,  
	title    = {#3},
	#1,
}

\newtcolorbox{donebox}[3][]
{
	enhanced jigsaw,
	breakable,
	hbox,
	left=1pt,
	right=1pt,
	top=1pt,
	bottom=1pt,
	colframe = #2!50,
	borderline={0.5mm}{0mm}{black,dashed},
	colback  = #2!50,
	coltitle = #2!20!black,  
	title    = {#3},
	#1,
}

\ifdefined\CommentsOn 
	
	\def \msboxcolor {Cerulean}
	\def \mstextinboxcolor {white}
	\def \ms{\ctext{\msboxcolor}{\mstextinboxcolor}{\mbox{[\msname]}}}
	\newcommandtwoopt{\mstodo}[2][to do][]{
		\begin{todobox}{\msboxcolor}{}
			\begin{varwidth}{\linewidth}
				\textcolor{\mstextinboxcolor}{\texttt{#2$\blacktriangleright$\ms: #1}}
			\end{varwidth}
		\end{todobox}
	}
	\newcommandtwoopt{\msdone}[2][to do][]{
		\begin{donebox}{black}{}
			\begin{varwidth}{\linewidth}
				\textcolor{white}{\texttt{#2$\blacktriangleright$\ms: #1}}
			\end{varwidth}
		\end{donebox}
	}
	\newcommand{\msdraft}[1]{
		\draftbox{\ms}{#1}
	}
\else 
	
	\def \msboxcolor {}
	\def \mstextinboxcolor {}
	\def \ms{}
	\newcommandtwoopt{\mstodo}[2][to do][]{}
	\newcommandtwoopt{\msdone}[2][to do][]{}
	\newcommand{\msdraft}[1]{}
\fi 

\ifdefined\CommentsOn 

\def \arturboxcolor {PineGreen}
\def \arturtextinboxcolor {white}
\def \artur{\ctext{\arturboxcolor}{\arturtextinboxcolor}{\mbox{[\arturname]}}}
\newcommandtwoopt{\arturtodo}[2][to do][]{
	\begin{todobox}{\arturboxcolor}{}
		\begin{varwidth}{\linewidth}
			\textcolor{\arturtextinboxcolor}{\texttt{#2$\blacktriangleright$\artur: #1}}
		\end{varwidth}
	\end{todobox}
}
\newcommandtwoopt{\arturdone}[2][to do][]{
	\begin{donebox}{black}{}
		\begin{varwidth}{\linewidth}
			\textcolor{white}{\texttt{#2$\blacktriangleright$\artur: #1}}
		\end{varwidth}
	\end{donebox}
}
\newcommand{\arturdraft}[1]{
	\draftbox{\artur}{#1}
}
\else 

\def \arturboxcolor {}
\def \arturtextinboxcolor {}
\def \artur{}
\newcommandtwoopt{\arturtodo}[2][to do][]{}
\newcommandtwoopt{\arturdone}[2][to do][]{}
\newcommand{\arturdraft}[1]{}	
\fi 

\ifdefined\CommentsOn 

\def \frankboxcolor {OrangeRed}
\def \franktextinboxcolor {white}
\def \frank{\ctext{\frankboxcolor}{\franktextinboxcolor}{\mbox{[\frankname]}}}
\newcommandtwoopt{\franktodo}[2][to do][]{
	\begin{todobox}{\frankboxcolor}{}
		\begin{varwidth}{\linewidth}
			\textcolor{\franktextinboxcolor}{\texttt{#2$\blacktriangleright$\frank: #1}}
		\end{varwidth}
	\end{todobox}
}
\newcommandtwoopt{\frankdone}[2][to do][]{
	\begin{donebox}{black}{}
		\begin{varwidth}{\linewidth}
			\textcolor{white}{\texttt{#2$\blacktriangleright$\frank: #1}}
		\end{varwidth}
	\end{donebox}
}
\newcommand{\frankdraft}[1]{
	\draftbox{\frank}{#1}
}
\else 

\def \frankboxcolor {}
\def \franktextinboxcolor {}
\def \frank{}
\newcommandtwoopt{\franktodo}[2][to do][]{}
\newcommandtwoopt{\frankdone}[2][to do][]{}
\newcommand{\frankdraft}[1]{}
\fi 

\ifdefined\CommentsOn 

\def \sophiaboxcolor {Purple}
\def \sophiatextinboxcolor {white}
\def \sophia{\ctext{\sophiaboxcolor}{\sophiatextinboxcolor}{\mbox{[\sophianame]}}}
\newcommandtwoopt{\sophiatodo}[2][to do][]{
	\begin{todobox}{\sophiaboxcolor}{}
		\begin{varwidth}{\linewidth}
			\textcolor{\sophiatextinboxcolor}{\texttt{#2$\blacktriangleright$\sophia: #1}}
		\end{varwidth}
	\end{todobox}
}
\newcommandtwoopt{\sophiadone}[2][to do][]{
	\begin{donebox}{black}{}
		\begin{varwidth}{\linewidth}
			\textcolor{white}{\texttt{#2$\blacktriangleright$\sophia: #1}}
		\end{varwidth}
	\end{donebox}
}
\newcommand{\sophiadraft}[1]{
	\draftbox{\sophia}{#1}
}
\else 

\def \sophiaboxcolor {}
\def \sophiatextinboxcolor {}
\def \sophia{}
\newcommandtwoopt{\sophiatodo}[2][to do][]{}
\newcommandtwoopt{\sophiadone}[2][to do][]{}
\newcommand{\sophiadraft}[1]{}
\fi 

\ifdefined\CommentsOn 

\def \allboxcolor {yellow}
\def \alltextinboxcolor {black}
\def \all{\ctext{\allboxcolor}{\alltextinboxcolor}{\mbox{[\allname]}}}
\newcommandtwoopt{\alltodo}[2][to do][]{
	\begin{todobox}{\allboxcolor}{}
		\begin{varwidth}{\linewidth}
			\textcolor{\alltextinboxcolor}{\texttt{#2$\blacktriangleright$\all: #1}}
		\end{varwidth}
	\end{todobox}
}
\newcommandtwoopt{\alldone}[2][to do][]{
	\begin{donebox}{black}{}
		\begin{varwidth}{\linewidth}
			\textcolor{white}{\texttt{#2$\blacktriangleright$\all: #1}}
		\end{varwidth}
	\end{donebox}
}
\newcommand{\alldraft}[1]{
	\draftbox{\ab}{#1}
}
\newcommand\See[1]{\marginpar{\tiny#1$\downarrow$}}
\newcommand\Here[1]{\marginpar{\tiny#1$\uparrow$}}
\else 

\def \allboxcolor {}
\def \alltextinboxcolor {}
\def \all{}
\newcommandtwoopt{\alltodo}[2][to do][]{}
\newcommandtwoopt{\alldone}[2][to do][]{}
\newcommand{\alldraft}[1]{}
\newcommand\See[1]{}
\newcommand\Here[1]{}
\fi 

\begin{document}
\title{A Multi-Agent Model for Opinion Evolution under Cognitive Biases }
%
%
\author{M\'{a}rio S. Alvim \inst{1}
\and
Artur Gaspar da Silva\inst{1}
\and
Sophia Knight\inst{2}
\and
Frank  Valencia\inst{4,5}
\thanks{M\'{a}rio S.\ Alvim and Artur Gaspar da Silva were partially supported by CNPq, CAPES, and FAPEMIG. Frank Valencia's contribution to this work is partially supported by the SGR project PROMUEVA (BPIN 2021000100160) supervised by Minciencias.}
}
\authorrunning{Alvim et al.}
%
\institute{
Department of Computer Science, UFMG, Brazil \and
Department of Computer Science, University of Minnesota Duluth, USA
\and
CNRS-LIX, \'{E}cole Polytechnique de Paris, France \and
Pontificia Universidad Javeriana Cali, Colombia
}
\maketitle             
%
\begin{abstract}

    We generalize the DeGroot model for opinion dynamics to better capture realistic social scenarios. We introduce a model where each agent has their own individual \emph{cognitive biases}. Society is represented as a directed graph whose edges indicate how much agents influence one another. Biases are represented as the functions in the square region $[-1,1]^2$ and categorized into four sub-regions  based on the potential reactions they may elicit in an agent during instances of \emph{opinion disagreement}. Under the assumption that each bias of every agent is a \emph{continuous} function within the region of receptive but resistant reactions ($\Rec$), we show that the society converges to a consensus if the graph is strongly connected. Under the same assumption, we also establish that the entire society converges to a unanimous opinion if and only if the \emph{source components} of the graph—namely, strongly connected components with no external influence—converge to that opinion. We illustrate that convergence is not guaranteed for strongly connected graphs when biases are either  discontinuous functions in  $\Rec$ or not  included in $\Rec$.  We showcase our model through a series of examples and simulations, offering insights into how opinions form in social networks under  cognitive biases.

    \keywords{Cognitive bias, Multi-Agent Systems, Social Networks}
\end{abstract}

\section{Introduction}
\label{sec:introduction}

In recent years, the significance and influence of social networks have experienced a remarkable surge, capturing widespread attention and shaping users' opinions in substantial ways. The  \emph{dynamics of opinion/belief formation} in social networks involves individuals expressing their opinions, being exposed to the opinions of others, and adapting or reinforcing their own views based on these interactions. Modeling these dynamics allows us to gain insights into how opinions form, spread, and evolve within social networks.

The DeGroot multi-agent model \cite{degroot} is one of most prominent formalisms for opinion formation dynamics in social networks. Society is represented as a directed graph whose edges indicate how much individuals (called \emph{agents}) influence one another. Each agent has an opinion  represented as a value in $[0,1]$ indicating the strength of their agreement with an underlying proposition (e.g., ``\emph{vaccines are safe}''). They repeatedly update their opinions  with the weighted average of their opinion differences (level of \emph{disagreement}) with those who influence them. The DeGroot model is valued for its tractability, derived from its connection with matrix powers and Markov chains, and it remains a significant focus of study providing a comprehensive understanding of opinion evolution\cite{survey}. 

 Nevertheless, the DeGroot model  has an important caveat: It assumes \emph{homogeneity} and \emph{linearity} of opinion update. In social scenarios, however, two agents may update their opinions differently depending on their individual \emph{cognitive biases} on disagreement—i.e., how they interpret and react towards the level of disagreements with others. This results in more complex updates that may involve non-linear even non-monotonic functions. For example, an individual under \emph{confirmation (cognitive) bias} \cite{Aronson10} may ignore the opinion of those whose level of disagreement with them is over a certain threshold. In fact, much of the unpredictability in opinion formation is due to users' \emph{biases} in their belief updates, where users  sometimes tend to reinforce their original beliefs, instead of questioning and updating their opinions upon encountering new information. Indeed, rather than perfect rational agents, users are often subject to cognitive biases.

In an earlier FORTE paper \cite{alvim2021multi},
we introduced a DeGroot-like model with a \emph{non-linear} update mechanism tailored for a specific type of confirmation bias. The model was shown to be tractable and it provides insights into the effect of this cognitive bias in opinion dynamics. Nevertheless, it also assumes homogeneity of opinion update, and choosing a particular function to represent the bias, although natural, may seem somewhat ad-hoc.

To address the above-mentioned caveat, in this paper we introduce a generalization of the DeGroot model that allows for \emph{heterogeneous} and \emph{non-linear} opinion updates. Each agent has their own individual cognitive biases on levels of disagreement. These biases are represented as arbitrary functions in the square region $[-1,1]^2$. The model then unifies disparate belief update styles with bias into a single framework which takes \emph{disagreement} between agents as the central parameter. Indeed, standard cognitive biases of great importance in social networks such as backfire effect \cite{Nyhan:2010}, authority bias \cite{Psychology2}, and confirmation bias \cite{Aronson10}, among others, can be represented in the framework.

We classify the biases in $[-1,1]^2$ into four sub-regions ($\Mal,\Rec,\BF,\Ins$) based on the cognitive reactions they may cause in an agent during instances of \emph{opinion disagreement}. 
For example, agents that are malleable, easily swayed, exhibit fanaticism or prompt to follow authoritative figures can be modelled with biases in $\Mal$. Agents that are receptive to other opinions, but unlike malleable ones, can exhibit some skepticism to fully accepting them can be modelled with biases in the region $\Rec$. Individuals that become more extreme when confronted with opposing opinions can be modelled by biases in $\BF$. Finally insular agents can be modelled with the bias in $\Ins$.

\emph{Consensus} is a central property for opinion formation dynamics. Indeed the inability to converge to consensus is often a sign of a polarized society. In this paper we use the above-mentioned region classification to provide the following insightful theoretical results for consensus.

\begin{itemize}
 \item Assuming that each bias of every agent is a \emph{continuous} function in $\Rec$, the society converges to a consensus if that society is strongly connected. This implies that  a strongly connected society can converge to a consensus if its members are receptive but resistant to the opinions of others.
 \item Under the same assumption, we also establish that the entire society converges to a unanimous opinion if and only if the \emph{source components} of the graph, i.e.,  strongly connected components with no external influence, converge to that opinion. This implies that upon agreeing on an opinion, closed and potentially influential groups, can make all individuals converge to that opinion in a society whose members are receptive but resistant. 
 \item We show that convergence is not guaranteed for strongly connected graphs when biases are either  discontinuous functions in  $\Rec$ or not  included in $\Rec$.  
 \end{itemize}

 We also demonstrate our model with examples and computer simulations that provide  insights into opinion formation under cognitive biases. The open code for these simulations 
 can be found at \url{https://github.com/bolaabcd/polarization2}.

\section{An Opinion Model with Cognitive Biases}
\label{sec:model}

The DeGroot model \cite{degroot} is a well-known model  for social learning. In this formalism each individual (\emph{agent}) repeatedly updates their current opinion by averaging the opinion values of those who influence them. But one of its limitation is that the model does not provide a mechanism for capturing the \emph{cognitive biases} under which each individual may interpret and react to the opinion of others.

In this section we introduce a generalization  of the DeGroot model with a mechanism to express arbitrary cognitive bias based on opinion  disagreement. 

\subsection{Influence Graph.} In social learning models, a \emph{community/society} is typically represented as a directed weighted graph with edges between individuals (agents) representing the direction and strength of the influence that one carries over the other. This graph is referred to as the \emph{Influence Graph}. 

\begin{definition}[Influence Graph]\label{d1}\label{d5}
    An ($n$-agent) \emph{influence graph} is a directed weighted graph $G=(A,E,I)$ with $A=\{1,\ldots,n\}$ the vertices, $E\subseteq A\times A$  the edges, and $I:A\times A \to [0,1]$ a weight function s.t. $I(i,j)=0$ iff $(i,j)\notin E$.
\end{definition}

 The vertices in $A$ represent $n$ agents of a given community or network. The set of edges $E\subseteq A\times A$ represents the (direct) influence relation between these agents; i.e.,  $(i,j)\in E$ means that agent $i$ influences agent $j$. The value $I(i,j)$, for simplicity written $\I{i}{j}$, denotes the strength of the influence: $0$ means no influence and a higher value means stronger influence.  We use $A_i$ to denote the set $\{ i \ | \ (j,i) \in E \}$ of agents that have a direct influence over agent $i$. 

\begin{remark}
In contrast to \cite{alvim2021multi}, we do not require agents to have nonzero self-influence. Furthermore, since we do not require the sum of influences over a given agent to be 1 (unlike \cite{degroot}), we  will use the following notation for \emph{proportional influence} of $j$ over $i$: $\overline{\I{j}{i}} = \frac{\I{j}{i}}{\sum_{k \in A_i} I_{k,i}}$
 if $(j,i)\in E$, else $\overline{\I{j}{i}}=0$.

\end{remark}

\subsection{General Opinion Update.}\label{subsec:biasedBeliefEvolution}

Similar to the DeGroot-like models in \cite{survey}, we model the evolution of  agents' opinions about some  underlying \emph{statement} or \emph{proposition}. For example, such a proposition could be ``\emph{vaccines are unsafe},'' ``\emph{human activity has little impact on climate change},'', ``\emph{AI poses a threat to humanity}'',  or   ``\emph{Reviewer 2 is wonderful}''. 

The \emph{state of opinion} (or \emph{belief state}) of  all the agents is represented as a vector in $[0,1]^{|A|}$. If $B$ is a state of opinion,  $\B{i} \in [0,1]$ denotes the \emph{opinion} (\emph{belief}, or \emph{agreement}) value of agent $i \in A$ regarding the underlying proposition. If $\B{i}=0$, agent $i$ completely disagrees with the underlying proposition; if $\B{i}=1$, agent $i$ completely agrees with the underlying proposition. Furthermore, the higher the value of $\B{i}$, the stronger the agreement with such a proposition.

At each time unit $t\in\N$, every agent $i \in A$ updates their opinion. We shall use $B^t$ to denote the state of opinion at time $t\in\N.$ We can now define a general  DeGroot-like opinion model as follows.  

\begin{definition}[Opinion Model]
 An \emph{Opinion Model} is a tuple $(G,B^0,\mu_G)$ where $G$ is an $n$-agent influence graph, $B^0$ is the initial state of opinion, and  
 $\mu_G: [0,1]^{n} \rightarrow [0,1]^{n}$ is a state-transition function, called \emph{update function}. For every $t\in\N$, the state of opinion at time $t+1$ is given by $B^{t+1}=\mu_G(B^t)$. 
\end{definition}

The update functions can be used to express any deterministic and discrete transition from one opinion state to the next, possibly taking into account the influence graph. This work singles out and characterizes a meaningful family of update functions extending the basic DeGroot model with cognitive biases that are based on opinion \emph{disagreement}. Intuitively, these update functions specify the reaction of an agent to the opinion disagreements with each of  their influencers. To build some intuition, we first recall the update function of the {DeGroot model}.

Below we omit the index from the update function $\mu_G$ if no confusion arises.

\subsection{DeGroot Update}

The standard DeGroot model \cite{degroot} is obtained by the following update function:
\begin{equation}\label{degroot-upd1:eq}
\mu(B)_i = \sum_{j \in A_i}\overline{\I{j}{i}} \B{j}  
\end{equation}
for every $i\in A$. Thus, in the DeGroot model each agent updates their opinion by taking the  weighted average of the opinions of those who influence them. We can rewrite Eq.\ref{degroot-upd1:eq} as follows:  
\begin{equation}\label{degroot-upd2:eq}
\mu(B)_i =  \B{i} + \sum_{j \in A_i}\overline{\I{j}{i}} (\B{j} - \B{i}). 
\end{equation}

Notice that DeGroot update is \emph{linear} in the agents' opinions and can be expressed in terms of \emph{disagreement}: The opinion of every agent $i$ is updated taking into account the weighted average of their \emph{opinion disagreement} or  \emph{opinion difference} with those who influence them. 

Intuitively, if $j$ influences $i$, then $i$'s opinion would tend to move closer to $j$'s. The \emph{disagreement term} $(\B{j} - \B{i}) \in [-1,1]$ in  Eq.\ref{degroot-upd2:eq} realizes this intuition. If  $(\B{j} - \B{i})$ is a negative term in the sum, the disagreement can be thought of as contributing with a magnitude of $|\B{j} - \B{i}|$ (multiplied by $\overline{\I{j}{i}}$) to \emph{decreasing} $i$'s belief in the underlying proposition. Similarly, if $(\B{j} - \B{i})$ is positive, the disagreement contributes with the same magnitude but to \emph{increasing $i$'s} belief. 

\subsection{Disagreement-Bias Update}
Now we generalize DeGroot updates by defining a class of update functions that also allows for \emph{non-linear} updates, and for each agent to react differently to opinion disagreement with distinct agents. We capture this reaction by means of bias functions $\bias{i,j}:[-1,1]\to[-1,1]$, where $(j,i) \in E$, on opinion disagreement stating how the bias of $i$ towards the opinion of $j$, $\bias{i,j}$,  affects $i$'s new opinion. 

In the following definition we use the clamp function for the interval $[0,1]$ which is defined as   $[r]_0^1 = \min(\max(r,0),1)$ for any $r\in\R$.

\begin{definition}[Bias Update] Let  $(G,B^0,\mu_G)$ be an opinion model with $G=(A,E,I)$. The function $\mu_G$ is a \emph{(disagreement) bias update} if for every $i\in A$,  
\begin{equation}\label{bias-upd:eq}
\mu_G(B)_i = \left[ \B{i} + \sum_{j \in A_i}\overline{\I{j}{i}}\bias{i,j}(\B{j} - \B{i}) \right]^1_0
\end{equation}
where each $\bias{i,j}$ with $(j,i) \in E$, called \emph{the (disagreement) bias} from $i$ towards $j$, is an endo-function\footnote{The biases we wish to capture can be seen as distortions of disagreements,  themselves disagreements. It seems then natural to choose $[-1,1]$ as the domain and co-domain of the bias function.} on $[-1,1]$.  The model $(G,B^0,\mu_G)$ is a \emph{(disagreement) bias} opinion model if $\mu_G$ is a disagreement bias update function. 
\end{definition}

The clamp function $[\cdot]_0^1$ guarantees that the right-hand side of Eq.\ref{bias-upd:eq} yields a valid belief value (a value in $[0,1]$).  Intuitively, the function $\bias{i,j}$ represents the direction and magnitude of how  agent $i$ reacts to their disagreement $\B{j} - \B{i}$ with agent $j$. If  $\bias{i,j}(\B{j} - \B{i})$ is a negative term in the sum of Eq.\ref{bias-upd:eq}, then the bias of agent $i$ towards $j$ contributes with a magnitude  of $|\bias{i,j}(\B{j} - \B{i})|$ (multiplied by $\overline{\I{j}{i}}$) to \emph{decreasing} $i'$s belief in the underlying proposition. Conversely,  if  $\bias{i,j}(\B{j} - \B{i})$ is positive, it contributes to \emph{increasing} $i'$s belief with the same magnitude.

Below we identify some particular examples of the cognitive biases that can be captured with disagreement-bias opinion models.

\begin{example}[Some Cognitive Biases]\label{basic-bias:example}  Clearly, the classical DeGroot update function Eq.\ref{degroot-upd2:eq} can be recovered from Eq.\ref{bias-upd:eq}  by letting every bias $\bias{i,j}$ be the identity on disagreement: i.e.,  
 $\bias{i,j}=\mathtt{degroot}$ where $\mathtt{degroot}(x)=x$.
 
 \emph{Confirmation Bias.} We now illustrate some form of \emph{confirmation bias}\cite{Aronson10} where agents are more \emph{receptive} to opinions that are closer to theirs. An example of confirmation bias can be obtained by letting $\bias{i,j}=\mathtt{conf}(x)=x(1+\delta-|x|)/(1+\delta)$ for a very small  non-negative constant $\delta$.\footnote{The confirmation bias function from \cite{alvim2021multi} uses $\delta=0$} In the following plots and simulations we fix $\delta=1 \times 10^{-4}$. This bias causes $i$ to pay less attention to the opinion of $j$ as their opinion distance $|x|=|\B{j} - \B{i}|$ tends to 1. 
 
 \emph{Backfire Effect.} Let us now consider another important  cognitive bias called {backfire effect}\cite{Nyhan:2010}. Under this effect an agent strengthens their position of disagreement with another agent if their opinions are significantly distant. A form of backfire effect can be obtained by letting $\bias{i,j}=\mathtt{backf}$ where $\mathtt{backf}(x) = -x^3$. Notice that unlike the DeGroot update, this bias contributes to changing $i$'s opinion with a \emph{magnitude} of $|\mathtt{backf}(\B{j} - \B{i})|$ (multiplied by $\overline{\I{j}{i}}$) \emph{but in the opposite direction} of the opinion of $j.$ This  potentially makes the new opinion of agent $i$  \emph{more} distant from that of $j$. 

\emph{Authority Bias.} Another common cognitive bias in social networks is the {authority bias} \cite{Psychology2} under which individuals tend to blindly follow  authoritative or influential figures often to the extreme. Let $\bias{i,j}=\mathtt{fan}$ be the \emph{sign} function, i.e.,  $\mathtt{fan}(x)=x/|x| \mbox{ if } x\neq 0, \mbox{ otherwise } \mathtt{fan}(x)=0$. This bias illustrates a case of die-hard \emph{fanaticism} of $i$ towards $j$. Intuitively, when confronted with any disagreement $x=\B{j}-\B{i}\neq 0$,  this bias contributes to changing $i$'s opinion with the \emph{highest magnitude}, i.e., $|\bias{i,j}(x)|=1$, in the \emph{direction} of the opinion of $j$.

 Finally we illustrate a bias that, unlike the previous, causes agents to ignore opinions of others. We call it the \emph{insular} bias $\bias{i,j}=\mathtt{ins}$ and it is defined as the zero function $\mathtt{ins}: x\mapsto 0.$ \qed
\end{example}

\begin{figure}[ht]
    \centering
    \resizebox{0.5\textwidth}{!}{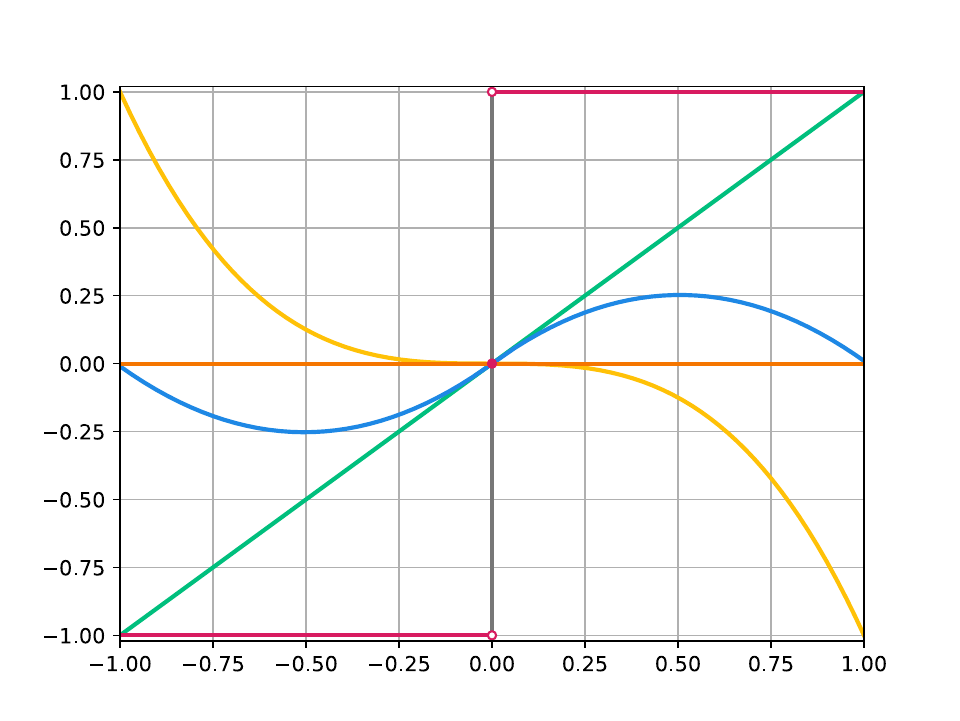}
    \caption{Bias functions from Ex.\ref{basic-bias:example} in the region $[-1,1]^2$:  $\mathtt{degroot}$ (in green),
$\mathtt{conf}$ (in blue), $\mathtt{backf}$ (in yellow), $\mathtt{fan}$ (in red), $\mathtt{ins}$ (in orange).
    \label{basic-bias:fig}}
\end{figure}

The particular bias function examples of Ex.\ref{basic-bias:example} are depicted in the \emph{square region} $[-1,1]^2$ in 
Fig.\ref{basic-bias:fig}. The functions may seem somewhat ad hoc but  in Section \ref{sec:biasAndConsensus} we identify a broad \emph{family} of \emph{bias functions} in the  region $[-1,1]^2$ that guarantees a property of central interest in multi-agent opinion evolution; namely, whether all the agents will converge to the same opinion, i.e. \emph{convergence to consensus}.

\begin{remark}
    We conclude this section by noting that unlike the DeGroot model, in Eq.\ref{bias-upd:eq} we allow agents to react with a distinct bias function to each of their influencers. This broadens the range of captured opinion dynamics and we illustrate this in the next section with an example exhibiting agents with different bias functions including those in Ex.\ref{basic-bias:example}. This, however, comes at a price; the update function can be non-linear in the agents' opinions (see e.g., functions $\mathtt{backf}$ and $\mathtt{conf}$). Thus, the analysis of opinion convergence using Markov chain theory for linear-system evolution as done for the DeGroot model is no longer applicable. In Section \ref{sec:biasAndConsensus} we study opinion convergence using  methods from real analysis.  
\end{remark}

\subsection{Vaccine Example}\label{subsec:examples}

Let us suppose that the proposition of interest is \emph{``vaccines are unsafe''}  
and  $G=(A,E,I)$ is as in Fig.\ref{fig:vaccineGraph}. Suppose that initally the agents $1,2,3$ are \emph{anti-vaxers} with opinion values $1.0,0.9,0.8$ about the proposition. In contrast, agents $4,5,6$ are initially \emph{pro-vaxers}, with opinion values $0.2,0.1,0.0$ about the proposition. Thus, the initial state of opinion is $B^0=(1.0,0.9,0.8,0.2,0.1,0.0)$. 

Notice that although agent $1$ is the most extreme anti-vaxer,  agent $6$, the most extreme vaxer, has the highest possible influence over them. As we shall illustrate below, depending on the bias of $1$ towards $6$, this may have a strong impact on the evolution of the opinion of agent $1$.

\begin{figure}
\centering
\resizebox{0.6\columnwidth}{!}{
    \begin{tikzpicture}
        \tikzstyle{every state}=[fill opacity=0.5,text opacity=1,thick,minimum size=12pt]
        
        \node[state, fill=blue1] (0) at (0,0) [] {1};
        \node[state, fill=blue2] (1) at (2,1) [] {2};
        \node[state, fill=blue3] (2) at (2,-1) [] {3};
        \node[state, fill=red1] (3) at (4,1) [] {4};
        \node[state, fill=red2] (4) at (4,-1) [] {5};
        \node[state, fill=red3] (5) at (6,0) [] {6};

        \draw
        (0) edge[<->] node{0.6} (1)
        (1) edge[->]  node{0.4} (3)
        (3) edge[->]  node{0.4} (5)
        (0) edge[->] node{0.4} (2)
        (2) edge[->]  node{0.6} (4)
        (4) edge[->]  node{0.6} (5)
        (3) edge[<->] node{0.2} (2)
        (5) edge[->, looseness=1.2, bend right=100]  node{1.0} (0)
        (0) edge[loop left] node{1.0} (0)
        (1) edge[loop above] node{1.0} (1)
        (2) edge[loop below] node{1.0} (2)
        (3) edge[loop above] node{1.0} (3)
        (4) edge[loop below] node{1.0} (4)
        (5) edge[loop right] node{1.0} (5)
        ;
    \end{tikzpicture}
    }
    \caption{Influence graph for vaccine example.}

    \label{fig:vaccineGraph}
\end{figure}
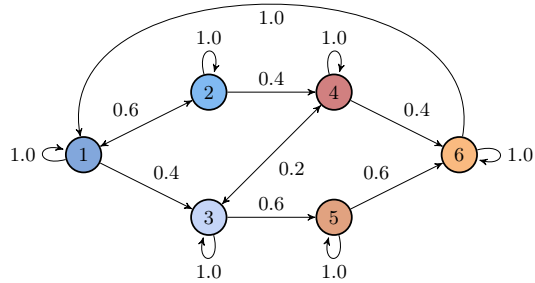

We now consider the evolution of their opinion under different update functions obtained by combining biases from Ex.\ref{basic-bias:example}.   In Fig.\ref{fig:vaccineMany} we show the evolution of opinions of vaxxers and anti-vaxxers using combinations of the bias functions from Fig.\ref{basic-bias:fig}.   Consider Fig.\ref{fig:vaccineCmany}. Agent $2$ reaches the extreme opinion value $1.0$ rather quickly because of their die-hard fanaticism towards the opinion of $1$ (i.e., $\bias{2,1}=\mathtt{fan}$). As the influence of $6$ on agent $1$ backfires ($\bias{1,6}=\mathtt{back}$), agent $1$ stays with belief value $1.0$. Eventually, all the other biases contribute to changing the belief value of the influenced agents towards $1.0$. Indeed, the agents converge to a consensus that vaccines are unsafe.

In Fig.\ref{fig:vaccineSDmany}, the influence of $3$ on agent $5$ backfires, since  $\bias{5,3}=\mathtt{backf}$. This makes their disagreement increase, moving agent $5$'s opinion closer to $0$. On the other hand, the opinion of agent $6$ is influenced at the same time by the belief values of $5$ and $4$ as in the DeGroot model $(\bias{6,5}=\bias{6,4}=\mathtt{degroot})$ so her opinion stays between theirs. 

Notice that in Fig.\ref{fig:vaccineFFmany} agent $5$ reacts to $3$ with die-hard fanaticism ($\bias{5,3}=\mathtt{fan}$) while $3$'s belief value does not converge to $0.0$ or $1.0$. Thus we obtain the looping behaviour of agent $5$. The fanaticism of agent $5$ propagates also to agent $6$ since he is influenced by agent $5$ by $\mathtt{degroot}$ bias. 

Finally, notice the behaviours in Fig.\ref{fig:vaccineMany} when all the agents have the same bias. In particular, Fig.\ref{fig:vaccineCB} suggests convergence to consensus when 
all the agents are under confirmation bias. In fact convergence to consensus is indeed guaranteed for this example as we shall later see in this paper.

The above illustrates that different types of biases can have strong impact on opinion evolution for a given influence graph. In the next section, we will identify meaningful families of bias as functions in the region $[-1,1]^2$.

\begin{figure}[ht]
\centering
\begin{subfigure}{0.3\linewidth}
    \resizebox{\textwidth}{!}{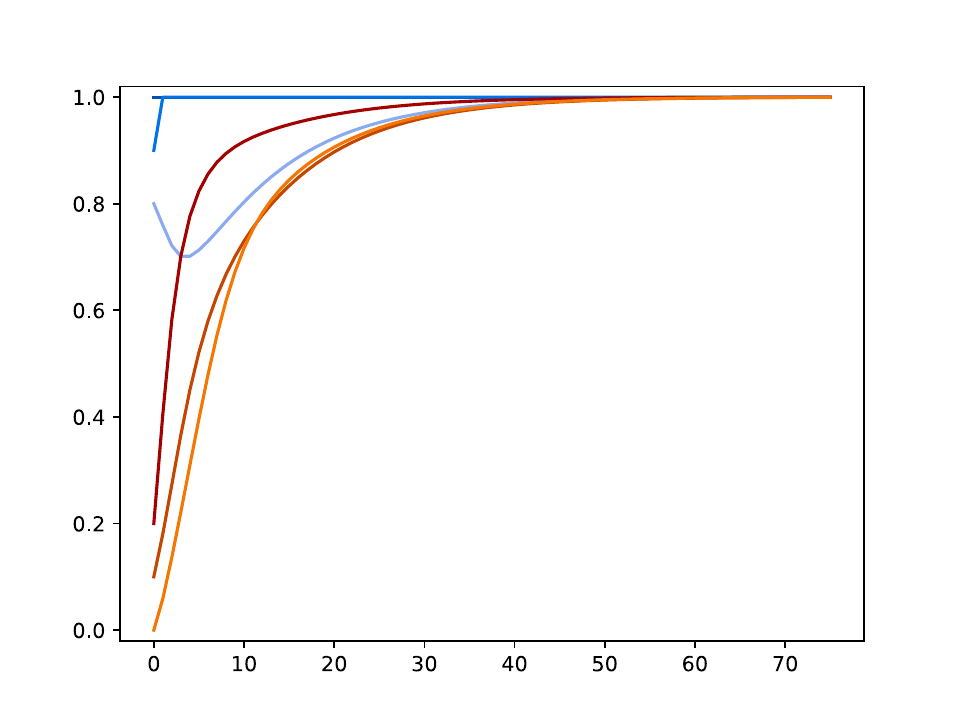}
    \caption{$\bias{2,1}=\mathtt{fan}$, $\bias{1,6}=\mathtt{backf}$, $\bias{4,2}=\bias{1,2}=\mathtt{degroot}$, otherwise $\bias{i,j}=\mathtt{conf}$.}
    \label{fig:vaccineCmany}
\end{subfigure}
\hfill 
\begin{subfigure}{0.3\linewidth}
    \resizebox{\textwidth}{!}{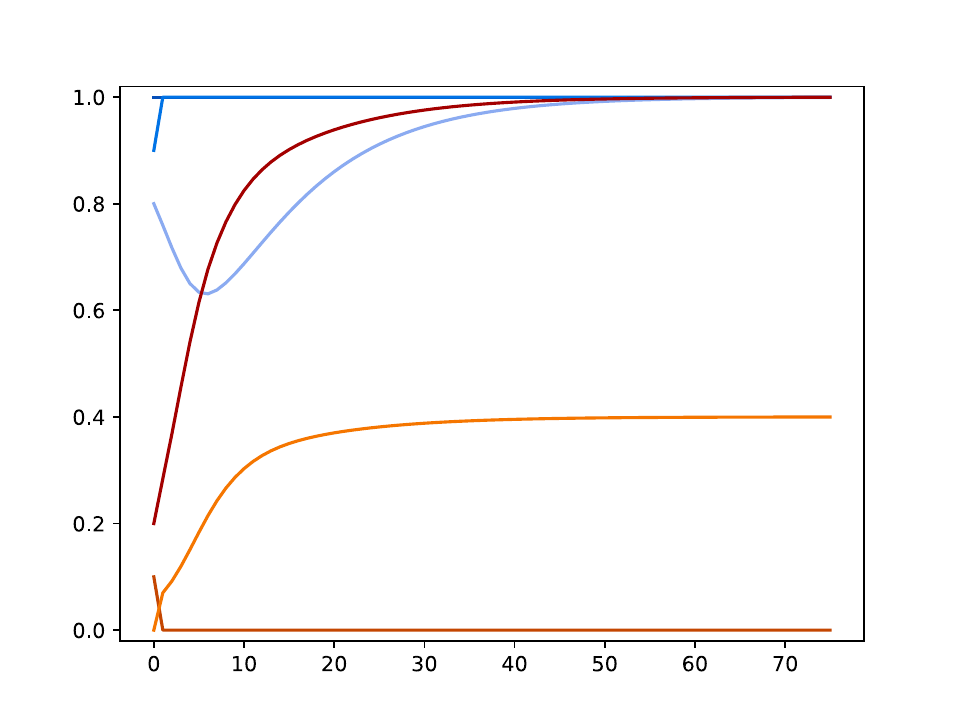}
    \caption{$\bias{2,1} = \mathtt{fan}$, $\bias{1,6}=\bias{5,3}=\mathtt{backf}$, $\bias{6,5}=\bias{6,4}=\mathtt{degroot}$, otherwise $\bias{i,j}=\mathtt{conf}$.}
    \label{fig:vaccineSDmany}
\end{subfigure}
\hfill 
\begin{subfigure}{0.3\linewidth}
    \resizebox{\textwidth}{!}{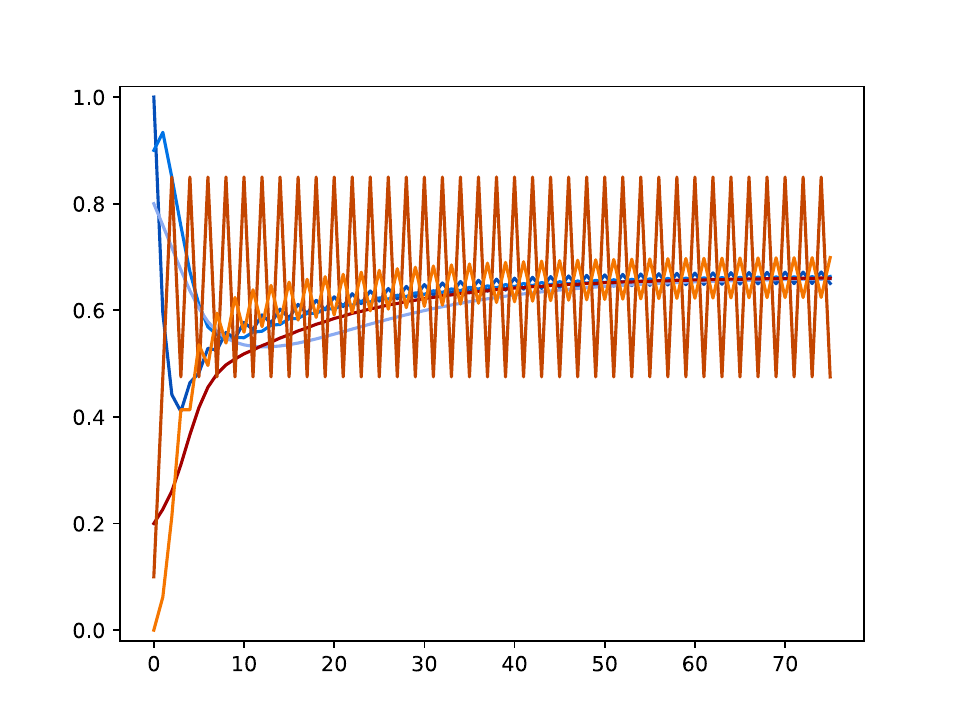}
    \caption{
     $\bias{5,3}=\mathtt{fan}$,  $\bias{4,3}=\mathtt{backf}$, $\bias{6,5}=\bias{1,6}=\mathtt{degroot}$, otherwise $\bias{i,j}= \mathtt{conf}$.}
     \label{fig:vaccineFFmany}
\end{subfigure}
\hfill 
\begin{subfigure}{0.3\linewidth}
    \resizebox{0.9\textwidth}{!}{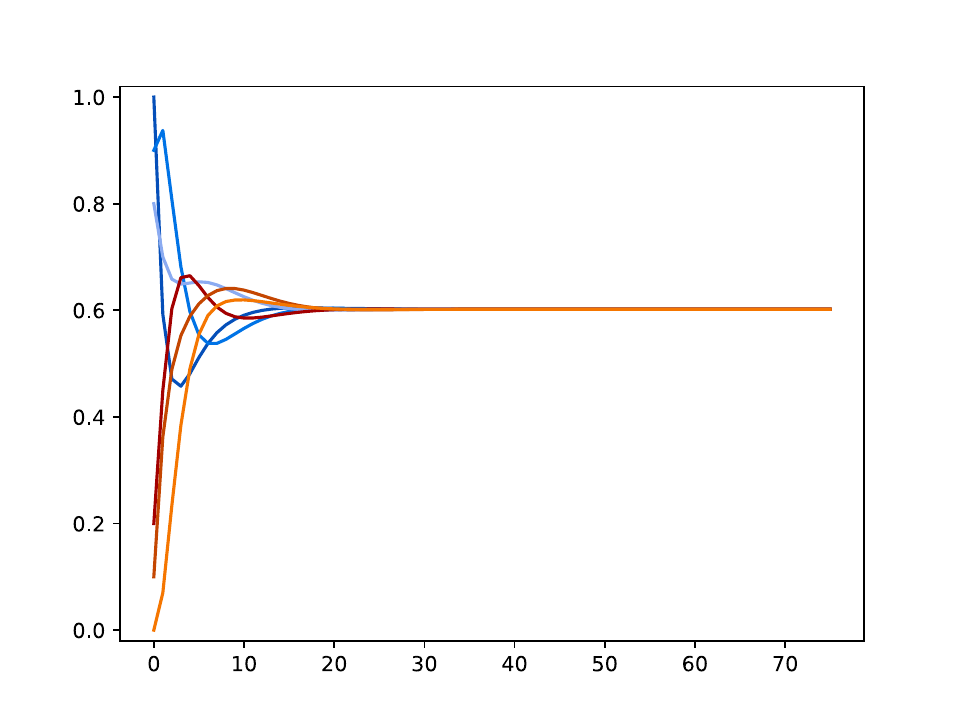}
    \caption{Each $\bias{i,j} = \mathtt{degroot}$.}
    \label{fig:vaccineDG}
\end{subfigure}
\hfill 
\begin{subfigure}{0.3\linewidth}
    \resizebox{0.9\textwidth}{!}{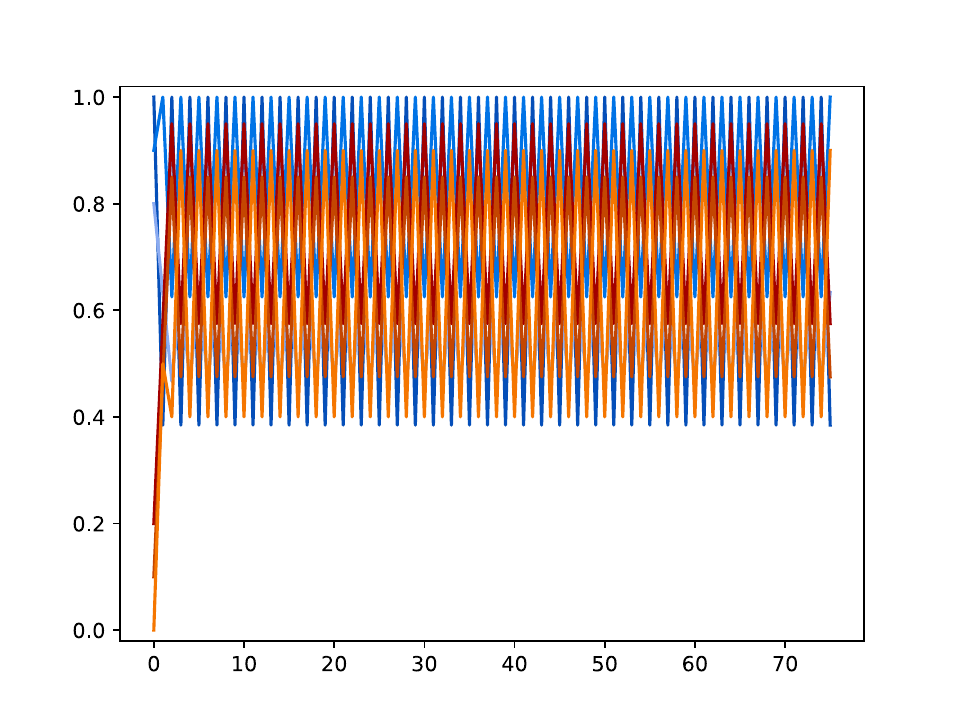}
    \caption{Each $\bias{i,j} = \mathtt{fan}$.}
    \label{fig:vaccineF}
\end{subfigure}
\hfill 
\begin{subfigure}{0.3\linewidth}
    \resizebox{0.9\textwidth}{!}{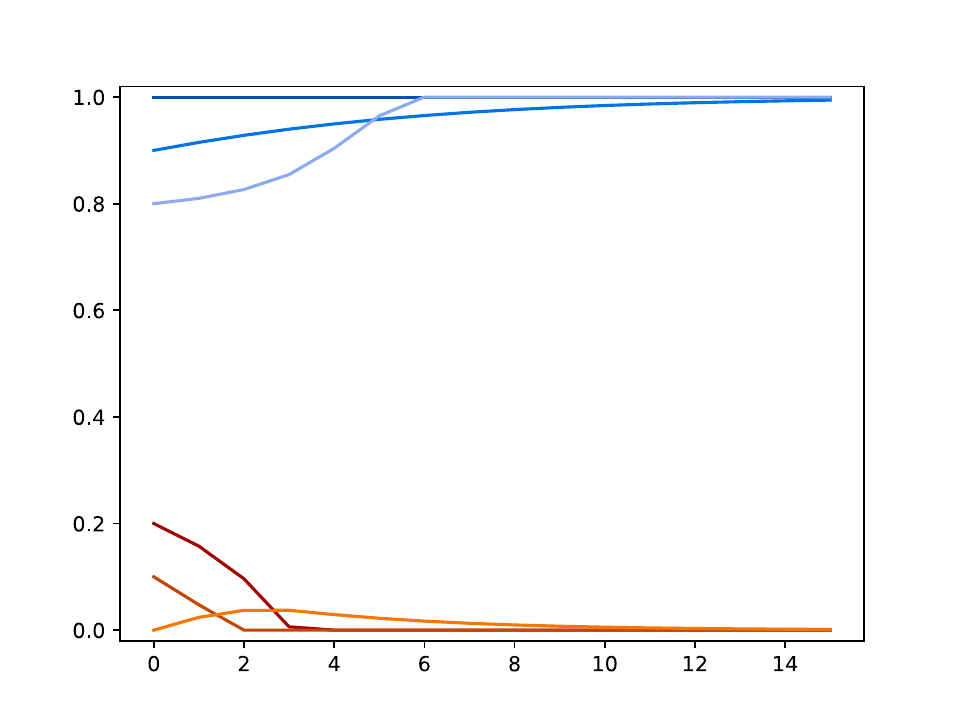}
    \caption{Each $\bias{i,j} = \mathtt{backf}$.}
    \label{fig:vaccineBF}
\end{subfigure}
\hfill 
\begin{subfigure}{0.3\linewidth}
    \resizebox{0.9\textwidth}{!}{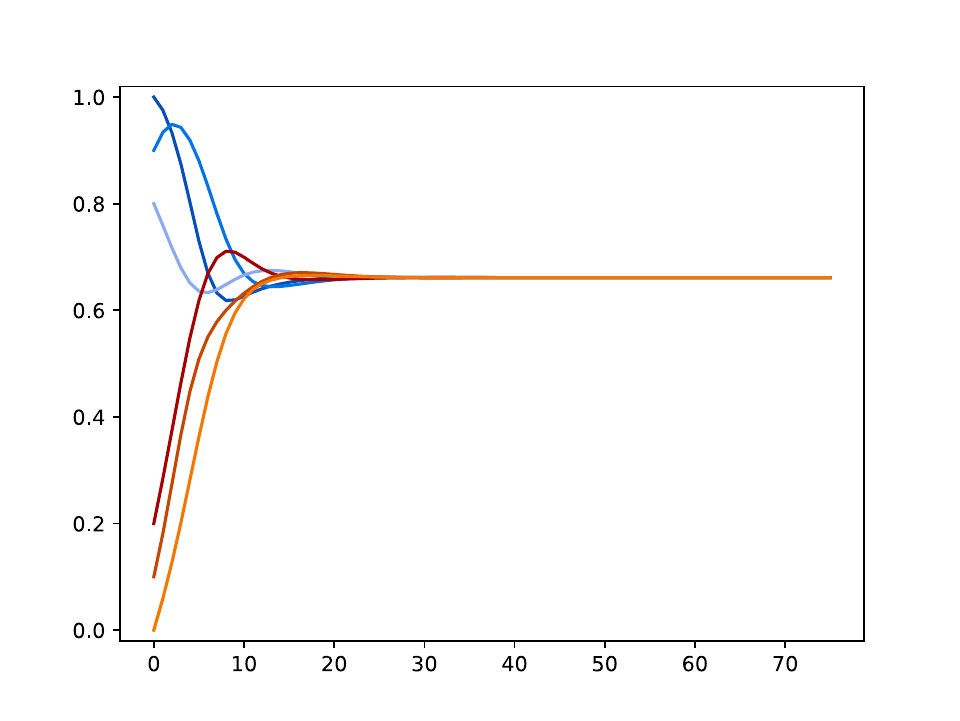}
    \caption{Each $\bias{i,j} = \mathtt{conf}$.}
    \label{fig:vaccineCB}
\end{subfigure}
\caption{Simulations for $G$ in Fig.\ref{fig:vaccineGraph} with  $B^0 = (1.0,0.9,0.8,0.2,0.1,0.0)$ using different biases. Each plot represents the evolution in time of the opinion of the agent in Fig.\ref{fig:vaccineGraph} with the same color.}\label{fig:vaccineMany}
\end{figure}

\section{Bias Region and Consensus}
\label{sec:biasAndConsensus}




%
%

%
%

Consensus is a property of central interest in social learning models. Indeed, failure to converge to a consensus is often an indicator of  polarization in a society. 

\begin{definition}[Consensus]   Let  $(G,B^0,\mu_G)$ be an opinion model with $G=(A,E,I)$.  We say that the subset of agents $A'\subseteq A$ \emph{converges to an opinion value $v\in[0,1]$} iff  for every $i \in A'$, $\lim_{t\to\infty}B^t_i = v.$ We say $A'\subseteq A$ \emph{converges to consensus} iff $A'$ converges to an opinion value $v$ for some $v$. 
 \end{definition}  

In this section we identify a broad and meaningful region of $[-1,1]^2$ where all the \emph{continuous} disagreement bias functions guarantee that agents converge to consensus under certain topological conditions on the influence graph.

\subsection{Bias Regions}\label{subsec:function_regions}
In what follows we say that a bias $\bias{i,j}$ \emph{is in a region} $R \subseteq [-1,1]^2$ if its function graph is included in $R$, i.e., if $\{ (x,\bias{i,j}(x)) \ | \ x \in [-1,1]\}\subseteq R$. We now identify regions of $[-1,1]^2$ that capture several notions of cognitive bias.


\begin{definition}[Bias Regions]\label{def:biasregions} Let $\SR$ be the square region $[-1,1]^2.$  Let the (sub)regions $\Mal,\Rec,\BF,\Ins \subseteq \SR$, named \emph{Malleability}, \emph{Receptive-Resistant}, \emph{Backfire} and \emph{Insular}, be  defined as follows:
\begin{align*} 
 \Mal =&  \{ (x,y)\in \SR \ |  
        (x< 0 \mbox{ and } y\leq x) \mbox{ or } (x> 0 \mbox{ and } y\geq x) \mbox{ or } x=0\}  \\
 \Rec =& \{ (x,y)\in \SR \ | \ (x< 0 \mbox{ and } x<y<0) \mbox{ or } (x>0 \mbox{ and } 0<y<x) \mbox{ or } x= y=0\}  \\ 
 \BF =& \{ (x,y) \in \SR \ | \ (x<0 \mbox{ and }  0<y) \mbox{ or } (x>0 \mbox{ and }   y<0)  \mbox{ or } x=y=0 \}  \\
 \Ins  =& \{ (x,y) \in \SR \ | \  y=0  \}. 
\end{align*}
\end{definition}

\begin{figure}[ht]
    \centering
    \resizebox{0.5\textwidth}{!}{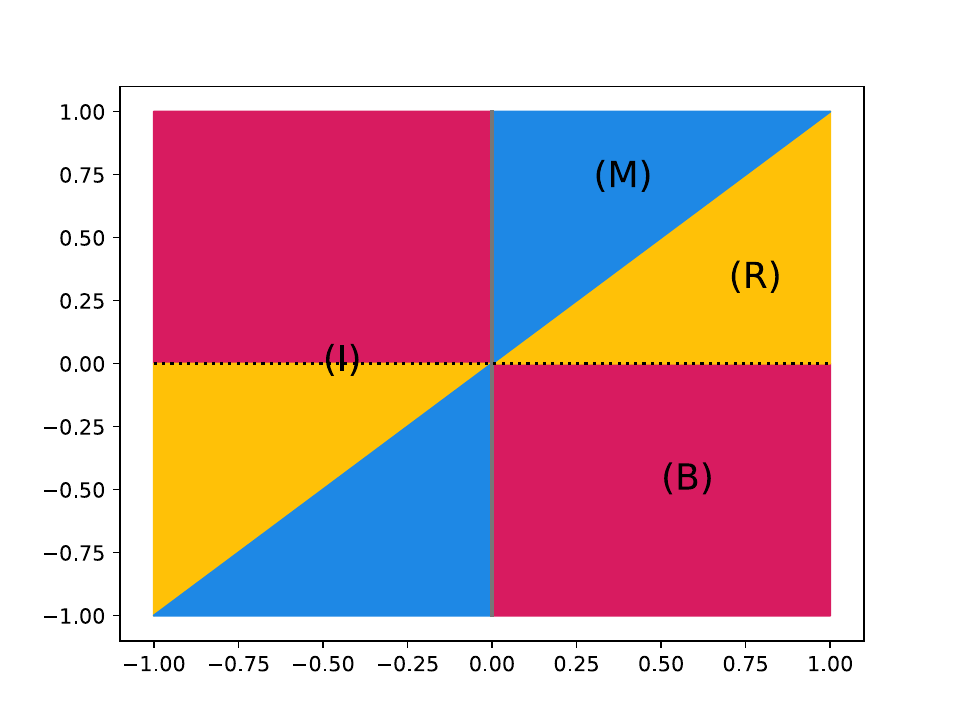}
    \caption{Bias Regions: Malleability ($\Mal$, in blue), Receptive-Resistant ($\Rec$, in yellow), Backfire  ($\BF$, in red), Insular ($\Ins$, the dotted line $y=0$).
    } 
    \label{fig:regions}
\end{figure}
The regions are  depicted in Fig.\ref{fig:regions}. Notice that if a point $(x,y)$ of a bias $\beta_{i,j}$ is in the  \emph{Malleability} region $\Mal$ (i.e., $y=\beta_{i,j}(x)$ and $(x,y)\in \Mal$) it means that for a disagreement $x=\Dif{j}{i}$
between $j$ and $i$, the bias will contribute with a magnitude $|y|\geq|x|$ (multiplied by $\overline{I_{j,i}}$) to changing the opinion of $i$ in the direction of $j$'s opinion. Since $|y|\geq|x|$,  
depending on the value of $\overline{I_{j,i}}$, the opinion of $i$ may move to match $j$'s opinion or even further (which can make $i$'s new opinion even more extreme than that of $j$). 
Individuals that blindly follow  authoritative or influential figures, easily swayed agents, fanaticism, among others, can be modelled by bias functions in this region. Indeed the function $\mathtt{fan}$ from Ex.\ref{basic-bias:example} is in $\Mal$ (see Fig.\ref{basic-bias:fig}). The identity bias function ${\mathtt{degroot}}$ is also in $\Mal.$

Like in the case above, if a point $(x,y)$ of a bias $\beta_{i,j}$ is in the \emph{Receptive-Resistant} region $\Rec$, it also means that for a disagreement $x=\Dif{j}{i} \neq 0$ between $j$ and $i$, the bias contributes to changing the opinion of $i$ in the direction of $j$'s opinion. Nevertheless, the magnitude of contribution is not as high as the previous case, namely it is $|y|$ with $|x|>|y|>0$. Individuals that are receptive to other opinions but, unlike malleable ones, may demonstrate some resistance, reluctance, or skepticism to fully accept them, can be modelled in this region. The confirmation bias function $\mathtt{conf}(x)=x(1+\delta-|x|)/(1+\delta)$ from Ex.\ref{basic-bias:example}, where $\delta>0$ is a very small constant, is in $\Rec$ (see Fig.\ref{basic-bias:fig}). 

In fact, it is worth noticing that for any constant $\delta>0$, the resulting bias function $\beta_{i,j}(x)=x(1+\delta-|x|)/(1+\delta)$ is in $\Rec$. In the limit, however, we have
$\lim_{\delta\to\infty}x(1+\delta-|x|)/(1+\delta)=x=
{\mathtt{degroot}(x)}$ which is not in $\Rec$ but in $\Mal$. Therefore,  $\delta$ could be viewed as a \emph{parameter of receptiveness}; the higher the value of $\delta$, the more receptive and less resistant agent $i$ is toward $j$'s opinion. In the limit,  agent $i$ is not resistant and behaves as a malleable agent towards $j$.

Contrary to the previous two cases, if a point $(x,y)$ of a bias $\beta_{i,j}$ is in the \emph{Back-Fire} region $\BF$, it means that for a disagreement $x=\Dif{j}{i} \neq 0$ between $j$ and $i$, the bias contributes to changing the opinion of $i$ \emph{but in the opposite  direction} of $j$'s opinion. This bias can then cause the disagreement between $i$ and $j$ to grow. Individuals that become more extreme when confronted with a different opinion can be modelled by bias functions in this region. Indeed, the function $\mathtt{backf}$ from Ex.\ref{basic-bias:example} is in $\BF$ (see Fig.\ref{basic-bias:fig}).

Finally, if a point $(x,y)$ of a bias $\beta_{i,j}$ is in the Insular region $\Ins$, it means that $y=0$, thus for a disagreement $x=\Dif{j}{i} \neq 0$ between $j$ and $i$, the bias causes $i$ to completely ignore the opinion of $j$. Individuals that are stubborn or closed-minded can be modelled with the function in this region. In fact,  the function $\mathtt{ins}$ from Ex.\ref{basic-bias:example} is the only function in $\Ins$ (see Fig.\ref{basic-bias:fig}).

We conclude this section with a proposition stating that we can dispense with the clamp function whenever all the bias functions are  in the $\Rec$ region. 

\begin{proposition}[Update with Bias in $\Rec$]\label{no-clamp:prop}
    Given a Bias Opinion Model $(G,B^0,\mu_G)$ with $G = (A,E,I)$, if for all $(a,b) \in E$ we have $\bias{b,a} \in \Rec$, then for all $B \in [0,1]^{|A|}$ and $i \in A$: 
$\mu_G(B)_i =  B_i + \sum_{j \in A_i}\overline{I_{j,i}}\bias{i,j}(B_j - B_i).$
\end{proposition}

The proof of this proposition can be found in the Appendix.

\subsection{Consensus under Receptiveness in Strongly Connected Graphs}\label{subsec:formal_results}

Our first main result states the convergence to consensus for strongly connected societies  when all bias functions are continuous and in the Receptive-Resistant Region defined in \ref{subsec:function_regions}.  We need some standard notions from graph theory.

Recall that a \emph{path} from $i$ to $j$ in $G=(A,E,I)$ is a sequence $i_0i_1\ldots i_m$ such that $i=i_0$, $j=i_m$ and $(i_0,i_1),(i_1,i_2),\ldots (i_{m-1},i_m)$ are edges in $E$. The graph $G$ is \emph{strongly connected} iff there is path from any agent to any other.  We can now state our first consensus result.   



\begin{theorem}[Consensus I]\label{main:theorem}
     Let  $(G,B^0,\mu)$ be a bias opinion model with a strongly connected graph $G = (A,E,I)$. Suppose that for every $(j,i)\in E$, $\bias{i,j}$ is a continuous function in $\Rec$. Then the set of agents $A$ converges to  consensus. 
\end{theorem}

Hence, the continuous bias functions in $\Rec$ guarantee consensus in strongly connected graphs, regardless of initial beliefs. Intuitively, the theorem says that a strongly connected community/society will converge towards consensus if its members are receptive but resistant to the opinions of others.    

Notice that the Vaccine Example in Sec.\ref{subsec:examples} with all agents  under confirmation bias satisfy the conditions of Th.\ref{main:theorem}, so their convergence to consensus is guaranteed. In fact, the opinion difference between any two agents grows smaller rather rapidly (Fig.\ref{fig:vaccineCB} illustrates this). In contrast, Fig.\ref{fig:slow} illustrates an example, with different a different bias also in $\Rec$, where the opinion difference grows smaller much slowly. But since such an example also satisfies the conditions of  Th.\ref{main:theorem}, convergence to consensus is guaranteed. 

\begin{figure}[ht]
\centering
\begin{subfigure}{0.4\linewidth}
    \centering
    \resizebox{\textwidth}{!}{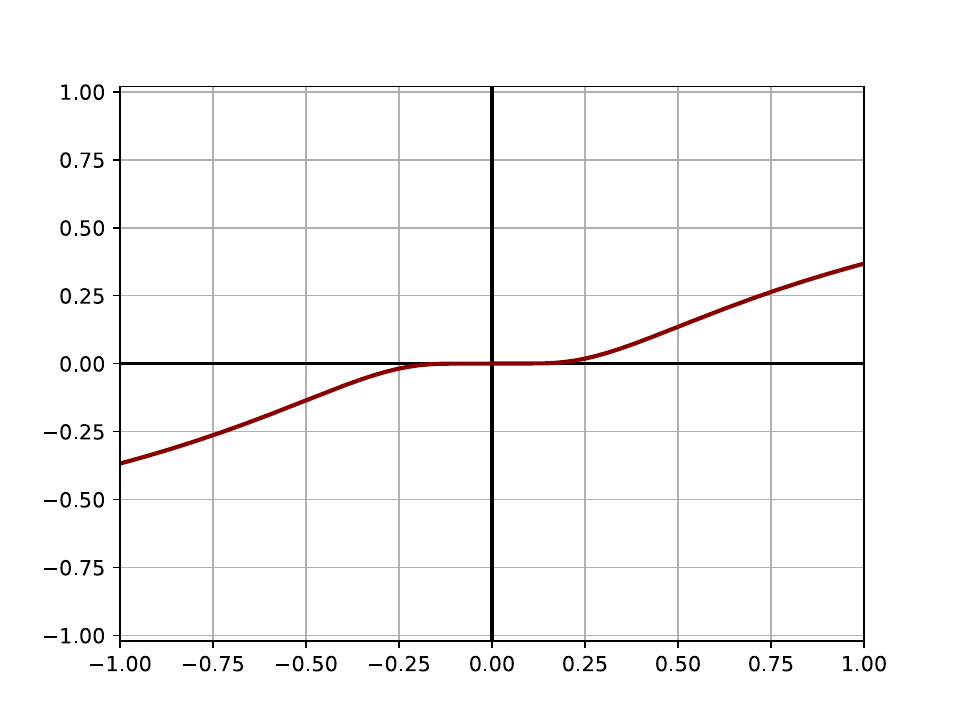}
    \caption{Bias function $\bias{1,2}(x) = \bias{2,1}(x) = \begin{cases} 
          0 & \text{if } x = 0 \\ 
          \frac{x}{|x|}\cdot e^{-\frac{1}{|x|}} &  \text{if } x \neq 0  
       \end{cases}$}\label{fig:funcSlow}
\end{subfigure}
\hfill 
\begin{subfigure}{0.4\linewidth}
    \centering
    \resizebox{\textwidth}{!}{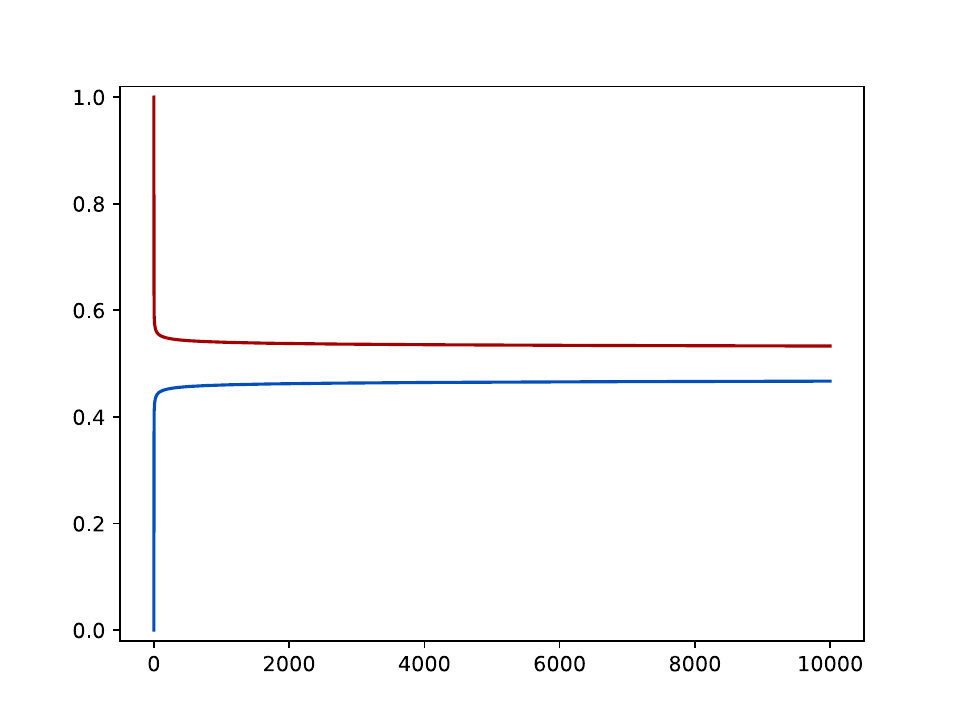}
    \caption{Each plot represents the evolution in time of the opinion of the agent in Fig.\ref{fig:twoagents} with the same color.}
\end{subfigure}
\caption{Simulations for $G$ in Fig.\ref{fig:twoagents} with $B^0 = (0.0, 1.0)$ using a bias function in region $\Rec$, with very slow convergence. Each plot represents the evolution in time of the opinion of the agent in Fig.\ref{fig:twoagents} with the same color.}\label{fig:slow}
\end{figure}

Before outlining the proof of this theorem, we elaborate on its conditions.
\paragraph{Discontinuous Bias.} Requiring continuity for the bias functions in Th.\ref{main:theorem} seems reasonable; small changes in an opinion disagreement value $x=B_j - B_i$ should result in small changes in $i$'s biased reaction to $x$. Nevertheless, if we relaxed
the continuity requirement, we would have the following counter-example. 

Consider a strongly connected graph with  two agents with $\I{1}{2} = \I{2}{1} = 1$, agent $1$ influences agent $2$ with the bias functions $\beta_{1,2} = \beta_{2,1}=f$, satisfying $f(x) = \frac{x}{2}$ if $x \in [-\frac{1}{2},\frac{1}{2}]$, $f(x) =\frac{x-0.5}{8}$ if $x \in (\frac{1}{2},1]$ and $f(x) = \frac{x+0.5}{8}$ if $x \in [-1,-\frac{1}{2})$. If one agent starts with belief value $1.0$ and the other $0.0$, then they will not converge to consensus (their belief values will approach $\frac{3}{4}$ and $\frac{1}{4}$, but will never reach those values).  Figure  \ref{fig:discontinuous} illustrates this example.

\begin{figure}[ht]
\centering
\begin{subfigure}{0.3\linewidth}
    \centering
    \resizebox{\textwidth}{!}{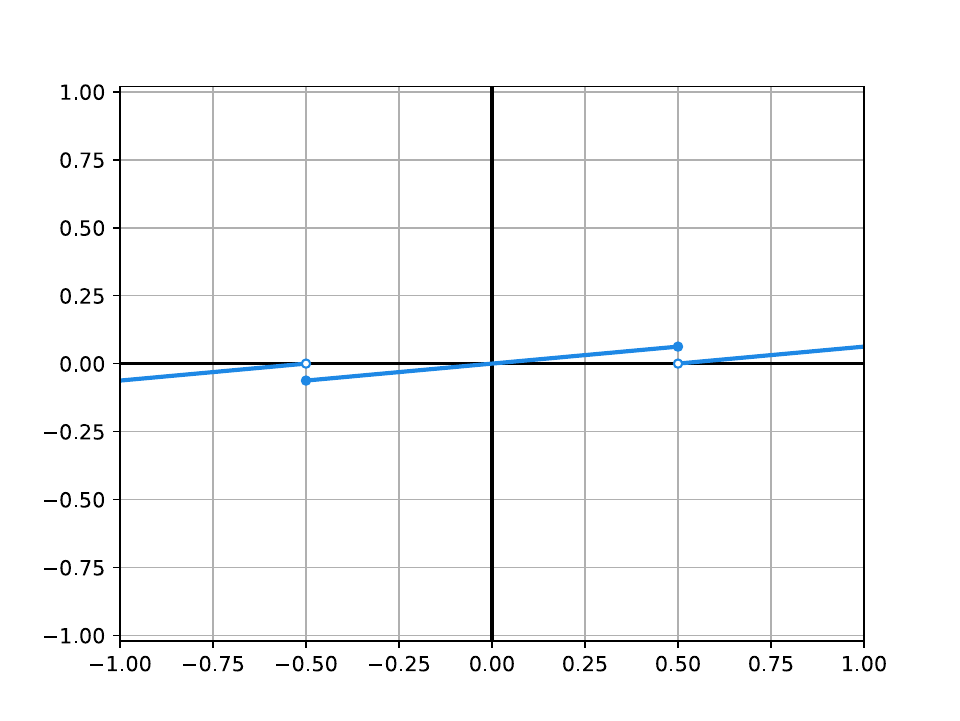}
    \caption{Bias function $\bias{1,2}(x) = \bias{2,1}(x) = \begin{cases} 
          \frac{x+0.5}{8} &  \text{if} -1 \leq x < -\frac{1}{2}  \\
          \frac{x}{8} & \frac{1}{2} \leq x \leq \frac{1}{2}  \\
          \frac{x-0.5}{8} & \text{if} \frac{1}{2} < x \leq 1 
       \end{cases}$}\label{fig:discontinuousfunc}
\end{subfigure}
\hfill 
\begin{subfigure}{0.3\linewidth}
    \centering

    \resizebox{0.3\columnwidth}{!}{%
    \begin{tikzpicture}
        \tikzstyle{every state}=[fill opacity=0.5,text opacity=1,thick,minimum size=12pt]
        \node[state, fill=blue1] (0) at (0,2) [] {1};
        \node[state, fill=red1] (1) at (0,4) [] {2};
        \draw
        (0) edge[<->] node{1.0} (1)
        ;
    \end{tikzpicture}
    }    
    \caption{Influence Graph ($I_{1,2} = I_{2,1} = 1.0$).}\label{fig:twoagents}
\end{subfigure}
\hfill 
\begin{subfigure}{0.3\linewidth}
    \centering
    \resizebox{\textwidth}{!}{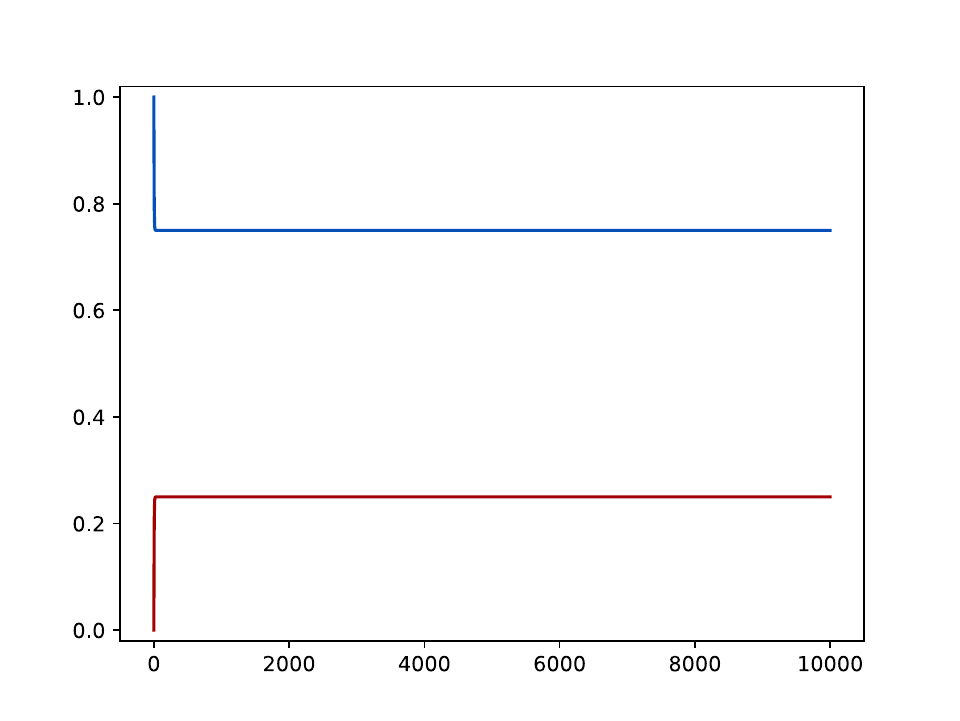}
    \caption{Each plot represents the evolution in time of the opinion of the agent in Fig.\ref{fig:twoagents} with the same color.}\label{fig:discontinuoushist}
\end{subfigure}
\caption{Counter-example to consensus for two agents with  non-continuous bias functions in $\Rec$ and with $B^0 = (1.0,0.0)$. }\label{fig:discontinuous}
\end{figure}

\begin{figure}[ht]
\centering
\begin{subfigure}{0.4\linewidth}
    \centering
    \resizebox{\textwidth}{!}{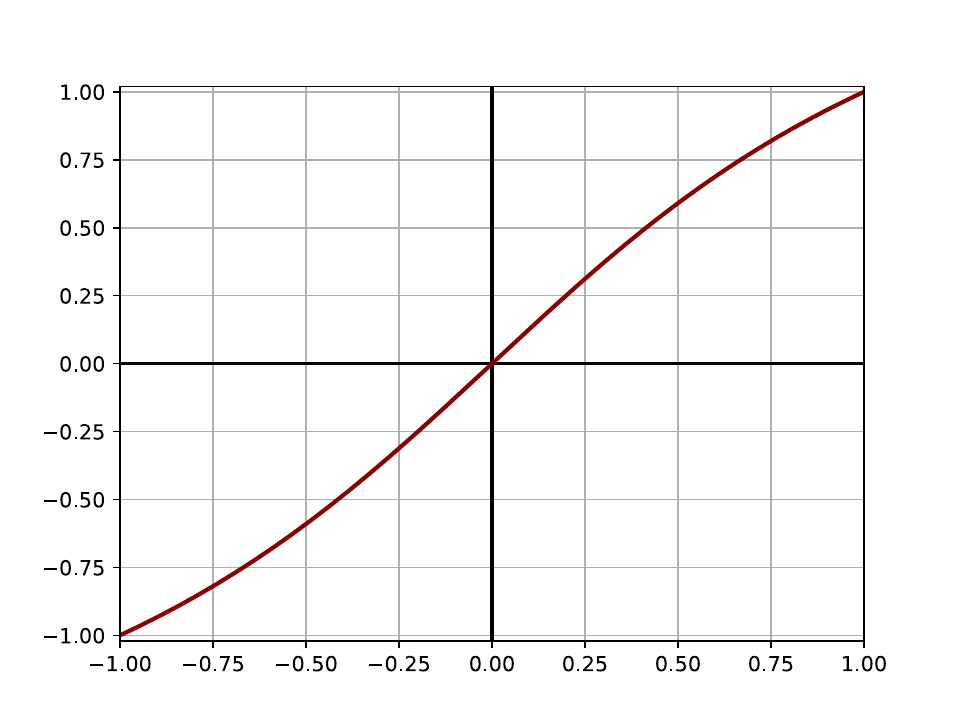}
    \caption{Bias function $\bias{1,2}(x) = \bias{2,1}(x) = \frac{\arctan x}{\arctan 1}$}\label{fig:puppetsimplerFunction}
\end{subfigure}
\hfill 
\begin{subfigure}{0.4\linewidth}
    \centering
    \resizebox{\textwidth}{!}{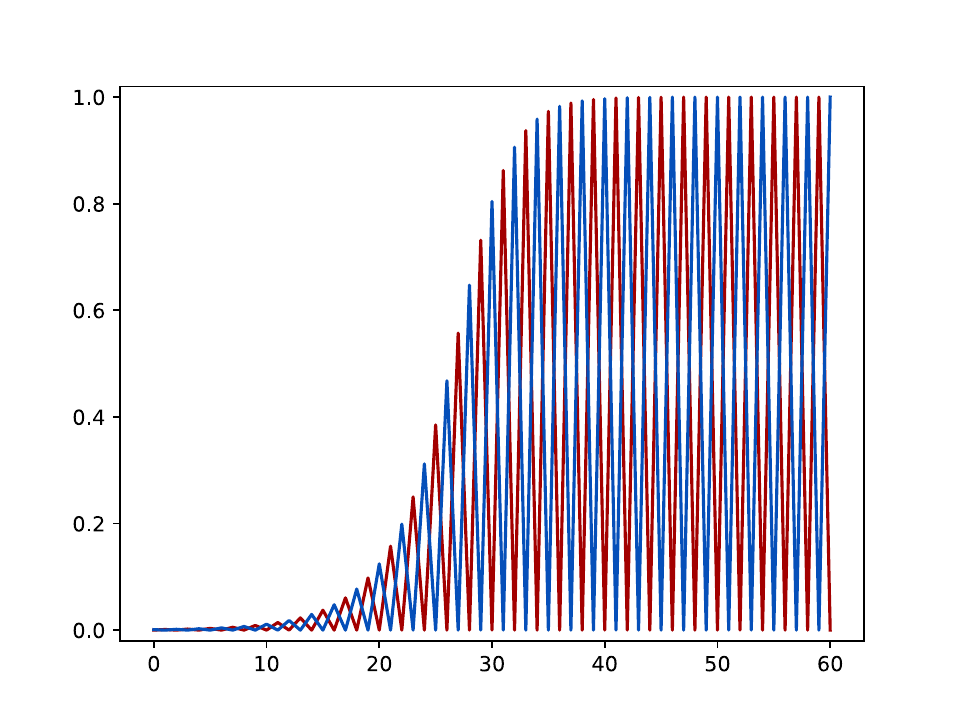}
    \caption{Each plot represents the evolution in time of the opinion of the agent in Fig.\ref{fig:twoagents} with the same color.}
\end{subfigure}
\caption{Counter-example for consensus when all bias function are continuous but allowed to have points in $\Mal$, with initial belief vector $B^0 = (0.0,0.001)$ and influence graph \ref{fig:twoagents}.}\label{fig:puppetsimpler}
\end{figure}

\paragraph{Bias Outside $\Rec$.} Notice that Th.\ref{main:theorem} requires bias functions to be in the responsive-resistant region $\Rec$.  We consider counter-examples where we allow bias functions outside this region in Th.\ref{main:theorem}.
If we allowed continuous bias functions outside $\Rec$ with points in the backfire region $\BF$, then the scenario in Fig.\ref{fig:vaccineBF} provides a counter-example to consensus. If we allow  continuous bias functions outside $\Rec$ with points in region $\Mal$, then the scenario in Fig.\ref{fig:puppetsimpler} is a counter-example to consensus: notice how the absolute value of their disagreement begins at $0.001$ and increases until it reaches $1$. Finally, it is clear that if we allowed the only function in $\Ins$, the insular bias, with the graph in Fig.\ref{fig:twoagents} and initial beliefs $B^0=(0,1)$, consensus will never be reached since the agents will ignore each other.

\subsection{Proof Outline of Th.\ref{main:theorem}.}
In this Section we outline
the proof of Th.\ref{main:theorem}. In the process we single out the central properties of the behaviour of agents that are receptive and yet resistant to disagreement. The complete proof can  can be found in the appendix. 

Let  $(G,B^0,\mu)$ be as in the statement of Th.\ref{main:theorem}.  Suppose $B=\mu^t(B^0)$ is the state of opinion at some time $t \geq 0$ where consensus has not yet been reached: i.e., assume $\min(B)\neq \max(B)$ where $\min(B)$ and $\max(B)$ are the minimum and maximum opinion values in $B$.  By assumption, all the biases $\bias{i,j}$ are in $\Rec$. Thus $\bias{i,j}(x)=y$, where $x=B_j - B_i$, contributes to update the opinion of $i$ in the direction of the opinion of $j$ but with a magnitude $|y|>0$ strictly smaller than $|x|$ if $|x|>0$ (or equal to $0$ if $|x|=0$). Using this and Prop.\ref{no-clamp:prop}\footnote{This follows from the known property that weighted averages of any set of values are always between the minimum and the maximum of those values.}, we show the new (updated) opinion of each $i$, $\mu(B)_i$, is bounded as follows:  
    
\begin{lemma}[Update Bounds]\label{lembounds} For each $i\in A$, $\min(B)\leq\mu(B)_i\leq\max(B)$. 
\end{lemma}

We use the above lemma to prove that the bounded sequences of minimum and maximum opinion values at each time, $\{\min(B^t)\}_{t \geq 0}$ and $\{\max(B^t)\}_{t \geq 0}$, are monotonically non-decreasing and monotonically non-increasing. Thus by the Monotone convergence theorem \cite{Sohrab:14}, they both converge. Therefore, by the Squeeze theorem \cite{Sohrab:14}, to prove Th.\ref{main:theorem}, it suffices to show that $\{\min(B^t)\}_{t \geq 0}$ and $\{\max(B^t)\}_{t \geq 0}$ converge to the same value. 

We first prove the following lemma which intuitively states that the number of extreme agents decrease with time. 

\begin{lemma}[Extreme Agents Reduction]\label{mainlemma} Suppose that $\min(B) \neq \max(B)$ and let  $M = \max(B)$. If $G$ has a path $i_1\ldots i_n$ such that $\B{i_n} = M$ and $\B{i_1} < M$, then $|\{ j \in A : B_j \geq M\}| > |\{ j \in A:\mu(B)_j \geq M\}|$. A symmetric property applies to the minimum.
\end{lemma}

To see the lemma's intuition, notice that  since  $G$ is strongly connected and $\min(B)<\max(B)$, $G$ indeed has a path $i_1\ldots i_n$ such that $\B{i_n} = M =\max(B)$ and $\B{i_1} < M$. In the path some agent $i_k$ whose belief value is equal to $M$ will be influenced by some agent with smaller belief value. Thus, since the bias functions are in $\Rec$, the opinion of $i_k$ will change in the direction of the smaller value, and thus will strictly decrease. Also, no agent that had a smaller  belief value will reach the current maximum, as the bias functions are in the region $\Rec$. 

Thus, because of  Lem.\ref{mainlemma} and $G$ being strongly connected, we conclude that the maximum (minimum) belief value will eventually decrease (increase). I.e.,   

\begin{corollary}\label{maincor}  Suppose that $\min(B) \neq \max(B).$ Then there exist $s,t>0$ such that $\max(\mu^s(B))<\max(B)$ and $\min(\mu^t(B))>\min(B).$
\end{corollary}

We now apply Bolzano-Weierstrass theorem \cite{Sohrab:14}\footnote{Every infinite bounded sequence in $\R^{n}$ has a convergent sub-sequence.} to find a sub-sequence $\{ \Bt{}{t}\}_{t \in \Delta}$ of $\{\Bt{}{t}\}_{t \in \N}$ that converges to some $\Bt{}{\infty}$. Notice that $\{ \max(\Bt{}{t})\}_{t \in \Delta}$ converges to $\max(\Bt{}{\infty})$ and it is a sub-sequence of the convergent sequence $\{ \max(\Bt{}{t})\}_{t \in \N}$, so $\{ \max(\Bt{}{t})\}_{t \in \N}$ should also converge to $\max (\Bt{}{\infty})$. Since  each bias function $\bias{i,j}$ is continuous, the update function $\mu$ is continuous. Therefore, $\{ \mu (\Bt{}{t})\}_{t \in \Delta}$ converges to $\mu (\Bt{}{\infty})$, and thus $\{\max(\mu (\Bt{}{t}))\}_{t \in \Delta}$ converges to $\max(\mu (\Bt{}{\infty}))$. But since the sequence $\{\max(\mu (\Bt{}{t}))\}_{t \in \Delta} = \{\max(\Bt{}{t+1})\}_{t \in \Delta}$ is a sub-sequence of the convergent sequence $\{ \max(\Bt{}{t})\}_{t \in \N}$, both must converge to the same value, hence $\max(\Bt{}{\infty}) = \max(\mu (\Bt{}{\infty}))$. Similarly, we can show that $\min(\Bt{}{\infty}) = \min(\mu (\Bt{}{\infty}))$. It can thus be shown that if we repeatedly apply $\mu $ to $\Bt{}{\infty}$, the maximum should not change, and the same applies to the minimum. More precisely, we conclude the following.  

\begin{corollary}
\label{maincor2} $\max(\Bt{}{\infty}) = \max(\mu^t (\Bt{}{\infty}))$ and   $\min(\Bt{}{\infty}) = \min(\mu^t (\Bt{}{\infty}))$ for each $t\geq 0$.
\end{corollary}


 Consequently,  if  $\min(\Bt{}{\infty}) \neq \max( \Bt{}{\infty})$ then Cor.\ref{maincor} and Cor.\ref{maincor2} lead us to a contradiction.  Therefore, $\min(\Bt{}{\infty}) = \max(\Bt{}{\infty})$ and thus, $\{\min(B^t)\}_{t \geq 0}$ and $\{\max(B^t)\}_{t \geq 0}$ converge to the same value $\max (\Bt{}{\infty})$ as wanted. \qed

\subsection{Consensus Under Receptiveness in Arbitrary Graphs}

Recall that Th.\ref{main:theorem} applies to strongly connected influence graphs. Our second main result  applies to arbitrary influence graphs. First we need to recall the notion of strongly-connected components of a graph. 

A \emph{strongly-connected component} of $G$ is a \emph{maximal subset} $S \subseteq A$ such that for each two $i,j\in S$, there is path from $i$ to $j$. A strongly-connected component $S$ is said to be a \emph{source component}  iff there is no edge $(i,j)\in E$ such that $i \in A\setminus S$ and $j\in S$. We use $\source(G)$ to denote the set of source components of $G$. 

Intuitively, a source component of a graph can be thought of as a closed group that is not  externally influenced but may influence individuals outside the group. The following theorem gives a characterization of consensus with biases in $\Rec$  for arbitrary graphs in terms of source components. 

\begin{theorem}[Consensus II]\label{consens2} Let  $(G,B^0,\mu)$ be a bias opinion model with $G=(A,E,I)$. Suppose that for every $(j,i)\in E$, $\bias{i,j}$ is a continuous function in $\Rec$.
Then the set of agents $A$ converges to consensus iff there exists $v\in[0,1]$ such that every source component $S \in  \source(G)$ converges to opinion $v.$
\end{theorem}

The above theorem, whose proof can be found in the Appendix, provides the following intuitive yet insightful remark. Namely, upon agreeing on an opinion, 
the closed and potentially influential groups, can make all individuals converge to that opinion in a society whose members are receptive but resistant. 

%
%


%
%

%
%
\section{Concluding Remarks and Related Work}
\label{sec:related_work}


We introduced a generalization of the DeGroot Model where agents interact under different biases. We identified the notion of bias on disagreement and made it the focus our model. This allowed us to identify families of biases that capture a broader range of social dynamics. We also provided theoretical results characterizing the notion of consensus for a broad family of cognitive biases.


 The relevance of biased reasoning in human interactions has been studied extensively in \cite{Psychology1}, \cite{Psychology2}, \cite{Psychology3}, and others. 
There is a great deal of work on formal models for belief change in social networks; we focus on the work on \emph{biased} belief update, which is the focus of this paper.  Some models were previously proposed to generalize the DeGroot model and introduce bias, for instance \cite{Generalize1}, \cite{Generalize2} and \cite{Generalize3} analyse the effects of incorporating a bias factor for each agent to represent biased assimilation: how much of the external opinions the agent will take into consideration. \cite{GeneralizeBF1} extends the model \cite{Generalize1} to include the effect of backfire-effect as well. The main difference of these models to our model is that biases are not incorporated in those models in terms of the disagreement level between agents, but either as an exponential factor that reduces the impact of the opinion of neighbours or by dynamically changing the weights of the DeGroot model. Thus, our model brings a new point of view to how distinct types of biases can be represented and identified. 

In \cite{mf}, it is proved that ``constricting'' update functions, roughly, functions where the extreme agents move closer to each other, lead to convergence in strongly connected social networks. This is similar to our theorem, indeed, the functions in our $\Rec$ region are easily shown to be constricting.  However, their social network model is more abstract than ours and further from real social networks, and they do not directly analyse biases as a function of disagreement.


%
%
%

\bibliographystyle{splncs04}
\bibliography{polar}

\appendix

\section{Appendix}

\begin{remark}[Fixed bias opinion model throughout the appendix]\label{rmkcaseA}
    For all our formal results we assume a fixed bias opinion model $(G,B^0,\mu)$ with $G = (A,E,I)$. Also, we assume an arbitrary $B \in [0,1]^{|A|}$.
\end{remark}

\begin{remark}[All bias functions are in $\Rec$]\label{rmkbiasA}
    For all our formal results, we assume that all bias functions are in region $\Rec$. Formally, this means that for all $x \in [-1,1]$ and $(j,i) \in E$, we have $0 > \bias{i,j}(x) > x$ if $x < 0$ and $0 < \bias{i,j}(x) < x$ if $x > 0$ and $\bias{i,j}(x) = 0$ if $x=0$. 
\end{remark}

\begin{remark}[All bias functions are continuous]\label{rmkcontA}
    Also, for all our results, we assume continuity of all bias functions. This means that for all $(j,i) \in E$ and for any sequence $\{x_n\}_{n \in \N}$ satisfying $\forall n \in \N: x_n \in [-1,1]$ and $\lim_{n \rightarrow \infty} x_n = L$ for some $L \in [-1,1]$, we have $\lim_{n\rightarrow \infty}\bias{i,j}(x_n) = \bias{i,j}(L)$.
\end{remark}

\begin{lemma}\label{aux1A}
    If $M \in [0,1]$ such that $M \geq \Bt{k}{}$ for any $k\in A$, then for any $(j,i) \in E$, if $\Bt{j}{} \neq \Bt{i}{}$, then $\Fdift{j}{i}{} < M-\Bt{i}{}$. Otherwise, $\Fdift{j}{i}{} \leq M-\Bt{i}{}$. A symmetrical property applies for the minimum.
\end{lemma}
\begin{proof}
    Let $(j,i) \in E$ be arbitrary. By Rmk.\ref{rmkbiasA}, if $\Dift{j}{i}{} < 0$, we have $\Fdift{j}{i}{} < 0 \leq M-\Bt{i}{}$, and if $\Dift{j}{i}{} > 0$, we have $\Fdift{j}{i}{} < \Dift{j}{i}{} \leq M-\Bt{i}{}$. This proves the first part of the lemma. If $\Dift{j}{i}{} = 0$, $\Fdift{j}{i}{} = 0 \leq M-\Bt{i}{}$ (because $(0,0)$ is the only point in $\Rec$ with $x=0$). This proves the second part of the lemma.

    The proof for the minimum is symmetrical.
\end{proof}

\begin{lemma}\label{boundnoclampA}
    Let $K_i = B_i + \sum_{j \in A_i}\overline{I_{j,i}}\bias{i,j}(B_j - B_i)$. Then for any $i \in A$, $\min(B) \leq  K_i \leq \max(B)$.
\end{lemma}
\begin{proof}
		\begin{align*}
			K_i &= \B{i}+\sum_{j\in A_i}\overline{I_{j,i}}\Fdift{j}{i}{t}\qquad \text{(by definition)}\\
			&\leq \B{i}+\sum_{j\in A_i}\overline{I_{j,i}}(\max(B)-\B{i})\qquad \text{(Lem.\ref{aux1A} with $M = \max(B)$)}
		\end{align*}
  If $A_i = \emptyset$, we have $K_i = \B{i}\leq \max(B)$. If not, we have $K_i\leq \B{i} + \max{B}-\B{i} = \max(B)$. We can find a symmetrical proof for the minimum ($\min(B) \leq K_i$).
\end{proof}

\begin{proposition}[Update with Bias in $\Rec$, Prop.\ref{no-clamp:prop} of section \ref{sec:biasAndConsensus}]\label{noclampA} For all $i \in A$:

\begin{equation*}
\mu(B)_i =  B_i + \sum_{j \in A_i}\overline{I_{j,i}}\bias{i,j}(B_j - B_i)
\end{equation*} 
\end{proposition}
\begin{proof}
    Let $i\in A$ be an arbitrary agent. By using Lem.\ref{boundnoclampA} (notice that $\mu(B)_i = \left[ K_i \right]^1_0$) and the fact that $B \in [0,1]^{|A|}$, we conclude that $\mu(B)_i = K_i$, which is exactly what the proposition says.
\end{proof}

\begin{lemma}[Update Bounds, Lem.\ref{lembounds} of Sec.\ref{sec:biasAndConsensus}]\label{lemboundsA}
    For each $i\in A$, $\min(B)\leq\mu(B)_i\leq\max(B)$.
\end{lemma}
\begin{proof}
    This follows directly from the combination of Prop.\ref{noclampA} and Lem.\ref{boundnoclampA}. 
\end{proof}

\begin{lemma}\label{extremA}
    There are $U,L \in [0,1]$ such that $\lim_{t\rightarrow \infty} \max(B^t) = U$ and $\lim_{t\rightarrow \infty} \min(B^t) = L$.
\end{lemma}
\begin{proof}
    It follows directly from Lem.\ref{lemboundsA} that $\max(B^t)$ and $\min(B^t)$ are respectively monotonically non-increasing and monotonically non-decreasing with respect to $t$. Then, for any integer $t \geq 0$ we have $\max(B^{t+1}) \leq \max(B^t)$ and $\min(B^{t+1}) \geq \min(B^t)$. Thus, the result follows from the monotone convergence theorem\cite{Sohrab:14}.
\end{proof}

\begin{lemma}\label{mucontA}
	The function $\mu$ is continuous.
\end{lemma}
\begin{proof}
	Note that we can compute $\lim_{Y\rightarrow B}\mu(Y)$ from it's individual values: let $i\in A$. Then, by Lem.\ref{noclampA}, we get $\mu(Y)_i = Y_i + \sum_{j\in A_i}\overline{I_{j,i}}\beta_{i,j}(Y_j-Y_i)$. Also, $Y\rightarrow B$ implies that $Y_i \rightarrow \B{i}$ for all $i\in A$. Then, $\lim_{Y\rightarrow B}\mu(Y)_i = \lim_{Y\rightarrow B}( Y_i+\sum_{j\in A_i}\overline{I_{j,i}}\beta_{i,j}(Y_j-Y_i)) = \B{i}+\sum_{j\in A_i}\overline{I_{j,i}}\lim_{Y\rightarrow B}\beta_{i,j}(Y_j-Y_i)$. Because all $\beta_{i,j}$ are continuous (Rem.\ref{rmkcontA}), this is equal to $\B{i}+\sum_{j\in A_i}\overline{I_{j,i}}\Fdif{j}{i} = \mu(B)$, thus concluding our proof, as $B$ is arbitrary. 
\end{proof}

\begin{lemma}[Extreme Agents Reduction, Lem.\ref{mainlemma} from Sec.\ref{sec:biasAndConsensus}]\label{mainlemmaA} Suppose that $\min(B) \neq \max(B)$ and let  $M = \max(B)$ If $G$ has a path $i_1\ldots i_m$ such that $\B{i_m} = M$ and $\B{i_1} < M$, then $|\{ j \in A : B_j \geq M\}| > |\{ j \in A:\mu(B)_j \geq M\}|$. A symmetrical property applies to the minimum.
\end{lemma}
\begin{proof}
	First note that by Lem.\ref{lemboundsA}, the maximum can not increase, which means that $\max(\mu(B)) \leq \max(B)$, and also by Lem.\ref{noclampA} the clamp function has no impact on the value of $\mu(B)$.

	Now note that if $\B{i} < \max(B)$, then if $A_i \neq \emptyset$, we can use Lem.\ref{aux1A} to conclude $(\mu(B))_i = \B{i} + \sum_{j\in A_i}\overline{I_{j,i}}\Fdif{j}{i} < \B{i}+\sum_{j\in A_i}\overline{I_{j,i}}(\max(B)-\B{i}) = \B{i} + (\sum_{j\in A_i} \overline{I_{j,i}})(\max(B)-\B{i}) = \B{i} + \max(B)-\B{i} = \max(B)$. Also, if $A_i = \emptyset$, then $(\mu(B))_i = \B{i} < \max(B)$. This means that the number of agents with belief value greater than or equal to $\max(B)$ cannot increase. We proceed to prove that at least one agent with belief $\max(B)$ will decrease it's belief value after the application of $\mu$.

	Let $i_1i_2i_3...i_m$ be the path declared in the assumption. Let $1\leq y\leq n$ be the smallest value such that $\B{i_y} = \max(B)$. Then $\B{i_{y-1}} < \max(B)$ by construction, and $A_y \neq \emptyset$ (because $y$ is influenced by $y-1$). We proceed to prove that $(\mu(B))_{i_y} < \max(B)$, which means $i_y$ is an agent with belief $\max(B)$ that will decrease it's belief after the application of $\mu$, which concludes our proof. Let $y-1 = u$

	\begin{align*}
		(\mu(B))_{i_y} &= \B{i_y} + \sum_{j\in A_{i_y}}\overline{I_{j,i_y}}\Fdif{j}{i_y}\\
		&= \B{i_y} + \overline{I_{i_u,i_y}}\Fdif{i_{u}}{i_y}+\sum_{j\in A_{i_y}\backslash \{i_{u}\}}\overline{I_{j,i_y}}\beta_{i_{y},j}(\B{j}-\B{i_y})\\
		&\leq \B{i_y} + \overline{I_{i_u,i_y}}\Fdif{i_{u}}{i_y}+(\sum_{j\in A_{i_y}\backslash \{i_{u}\}}\overline{I_{j,i_y}})(\max(B)-\B{i_y})\\
		&= \B{i_y} + \overline{I_{i_u,i_y}}\Fdif{i_{u}}{i_y}\\
		&< \B{i_y}\\
		&= \max(B)
	\end{align*}
    The proof for the minimum is symmetrical.
\end{proof}

\begin{remark}[Composition of $\mu$]\label{compMu}
    We consider that $\mu^0(B) = B$ and $\mu^k{B} = \mu^{k-1}(\mu(B))$ for any integer $k > 0$.
\end{remark}

\begin{lemma}\label{decmaxA}
    Suppose that $\min(B) \neq \max(B)$ and let  $M = \max(B)$. If $G$ has a path $i_1\ldots i_m$ such that $\B{i_m} = M$ and $\B{i_1} < M$, then $\max(\mu^{|A|-1}(B)) < \max(B)$. A symmetrical property applies to the minimum.
\end{lemma}
\begin{proof}
    Note that $B$ has at most $|A|-1$ agents with belief value $\max(B)$, otherwise $\min(B) = \max(B)$. Then, by Lem.\ref{mainlemmaA}, in $\mu(B)$ we have at most $|A|-2$ agents with belief value $\max(B)$. Let $k>0$ be the smallest integer such that $\mu^k(B)$ has no agent with belief $\max(B)$, and note that this means that $\mu^{k-1}(B)$ has at least one agent with belief $\max(B)$. We first prove that that $k\leq |A|-1$:

    We know $\mu^{k-1}(B)$ has at least one agent with belief $\max(B)$, and by Lem.\ref{mainlemmaA}, for all $k \geq y > 0$ we know that the number of agents in $\mu^{y-1}(B)$ with belief $\max(B)$ is at least one plus the number of agents with that belief in $\mu^y(B)$. This means that for all $k\geq z\geq 0$, $\mu^{z}(B)$ has at least $k-z$ agents with belief value $\max(B)$. If (for the sake of contradiction) $k\geq |A|$, $\mu^0(B) = B$ would have at least $k-0=k\geq |A|$ agents with belief greater than or equal to $\max(B)$, and because this would mean $B$ is consensus, we have a contradiction with $\max(B) \neq \min(B)$.

    Now note that, by Lem.\ref{lemboundsA} if $\mu^k(B)$ has no agent with belief $\max(B)$, then this means that $\max(\mu^k(B)) < \max(B)$, and it also means that for all $z \geq k$, $\max(\mu^z(B)) \leq \max(\mu^k(B))<\max(B)$. So, as we just proved $k \leq |A|-1$, $\max(\mu^{|A|-1}(B)) \leq \max(\mu^k(B)) < \max(B)$, as desired.

\end{proof}

\begin{remark}[Source Components]\label{dsources}
    A \emph{strongly-connected component} of $G$ is a \emph{maximal subset} $S \subseteq A$ such that for each two $i,j\in S$, there is path from $i$ to $j$. A strongly-connected component $S$ is said to be a \emph{source component}  iff no agent in $S$ is influenced by an agent outside $S$, i.e.,  iff there is no edge $(i,j)\in E$ such that $i \in A\setminus S$ and $j\in S$. From now on we use $\source(G)$ to be the set of source components of $G$. 
\end{remark}

\begin{theorem}[Consensus I, Thm.\ref{main:theorem} of Sec.\ref{sec:biasAndConsensus}]\label{main:theoremA}
     Let  $(G,B^0,\mu)$ be a bias opinion model with a strongly connected graph $G = (A,E,I)$. Suppose that for every $(j,i)\in E$, $\bias{i,j}$ is a continuous function in $\Rec$. Then the set of agents $A$ converges to  consensus. 
\end{theorem}
\begin{proof}
    By Bolzano–Weierstrass theorem\cite{Sohrab:14}, for all integers $t\geq 0$, as all values of $B^t$ are bounded in $[0,1]$, there is a convergent infinite subsequence of the sequence $\{B^t\}_{t\in\N}$. Let this be $\{B^t\}_{t\in \Delta}$, and let the value to which it converges be $B^\infty$. Now let $\Gamma^k = \{B^{t+k}\}_{t\in\Delta}$, and note that this is equal to $\{\mu^k(B^t)\}_{t\in \Delta}$. By Lem.\ref{mucontA}, this sequence converges to $\mu^k(B^\infty)$, because the composition of continuous functions is also continuous. We will use $k = |A|-1$ in this theorem.

    Let $\lim_{t\rightarrow\infty}\max(B^t) = U$ (this limit exists because of Lem.\ref{extremA}), and, for the sake of contradiction, suppose $B^\infty$ is not consensus. Then, note that, as assume that the graph is strongly connected we can guarantee that the path assumption of Lem.\ref{decmaxA} is satisfied for $B=B^\infty$, with $i_1\in A$ such that $B_{i_1}^\infty = \min(B^\infty)$ and $i_m\in A$ such that $B_{i_m}^\infty = \max(B^\infty)$. Then, by Lem.\ref{decmaxA}, $\max(\mu^{|A|-1}(B^\infty)) < \max(B^\infty)$. 

    We can reach a contradiction of this using the fact that $\{\max(B^t)\}_{t \in \Delta}$ (which converges to $\max(B^\infty)$) is a subsequence of the convergent sequence $\{\max(B^t)\}_{t \in \N}$. As the latter converges to $U$, so does the former, so $U = \max(B^\infty)$. Also, $\{\mu^{|A|-1}(B^t)\}_{t \in \Delta}$ converges to $\mu^{|A|-1}(B^\infty)$ as $\mu$ is continuous. We also notice that $\{\max(\mu^{|A|-1}(B^t))\}_{t \in \Delta} = \{\max(B^{t+|A|-1})\}_{t \in \Delta}$ converges to the value $\max(\mu^{|A|-1}(B^\infty))$ and is a subsequence of the convergent sequence $\{\max(B^t)\}_{t \in \N}$, so both converge to the same value $\max(\mu^{|A|-1}(B^\infty)) = U$. So we conclude that $\max(\mu^{|A|-1}(B^\infty)) = U = \max(B^\infty)$ which is a contradiction with the inequality  $\max(\mu^{|A|-1}(B^\infty)) < \max(B^\infty)$. Thus, $B^\infty$ must be consensus.

    Let $v$ be the value such that for all $i \in A$, $B_i^\infty = v$. As $\{\max(B^t)\}_{t \in \Delta}$ (which converges to $\max(B^\infty) = v$) is a subsequence of the convergent sequence $\{\max(B^t)\}_{t \in \N}$, both converge to the same value, so $\lim_{t\rightarrow \infty} \max(B^t) = v$. The same applies for the minimum, so $\lim_{t\rightarrow \infty} \max(B^t) = v = \lim_{t\rightarrow \infty} \min(B^t)$, and as $\max(B^t) \leq B_i^t \leq \min(B^t)$ for all $i \in A$ (Lem.\ref{lemboundsA}), by the Squeeze Theorem\cite{Sohrab:14}, $A$ converges to consensus.

\end{proof}

\begin{theorem}[Consensus II, Thm. \ref{consens2} of Sec.\ref{sec:biasAndConsensus}]\label{consens2A} Let  $(G,B^0,\mu)$ be a bias opinion model with $G=(A,E,I)$. Suppose that for every $(j,i)\in E$, $\bias{i,j}$ is a continuous function in $\Rec$.
Then the set of agents $A$ converges to consensus iff there exists $v\in[0,1]$ such that every source component $S \in  \source(G)$ converges to opinion $v$.
\end{theorem}
\begin{proof}
    If there is no $v \in [0,1]$ such that all source components converges to $v$, then $A$ does not converge to consensus by definition.

    If there is such a $v \in [0,1]$, then we only need to show that the belief value of agents outside of source components converge to $v$. 

    By Bolzano–Weierstrass theorem\cite{Sohrab:14}, for all integers $t\geq 0$, as all values of $B^t$ are bounded in $[0,1]$, there is a convergent infinite subsequence of the sequence $\{B^t\}_{t\in\N}$. Let this be $\{B^t\}_{t\in \Delta}$, and let the value to which it converges be $B^\infty$. Now let $\Gamma^k = \{B^{t+k}\}_{t\in\Delta}$, and note that this is equal to $\{\mu^k(B^t)\}_{t\in \Delta}$. By Lem.\ref{mucontA}, this sequence converges to $\mu^k(B^\infty)$, because the composition of continuous functions is also continuous. We will use $k = |A|-1$ in this theorem.

    Let $\max(B^\infty) = U$, and for the sake of contradiction, suppose $B^\infty$ is not consensus. Then at least one of the following is satisfied: $\min(B^\infty) \neq v$, $\max(B^\infty) \neq v$. We will assume $\max(B^\infty) \neq v$, if that's not the case we can get a symmetric proof with $\min(B^\infty) \neq v$.

    Also note that, as we assume that the belief value of all agents in source components converge to $v$, then their belief values will also converge to $v$ in the subsequence $\{\Bt{i}{t}\}_{t \in \Delta}$. So if $i \in S$ for some $S \in \source(G)$, then $\Bt{i}{\infty} = v$. So can guarantee that the path assumption of Lem.\ref{decmaxA} is satisfied, with $i_1\in S$ for some $\source(G)$ (so $(\mu^l(B^\infty))_{i_1} = v$) and $i_m\in A$ such that $(\mu^l(B^\infty))_{i_m} = \max(\mu^l(B^\infty))$. Then, by Lem.\ref{decmaxA}, $\max(\mu^{|A|-1}(B^\infty)) < \max(B^\infty)$. 

    We can reach a contradiction of this using the fact that $\{\max(B^t)\}_{t \in \Delta}$ (which converges to $\max(B^\infty)$) is a subsequence of the convergent sequence $\{\max(B^t)\}_{t \in \N}$. As the latter converges to $U$, so does the former, so $U = \max(B^\infty)$. Also, $\{\mu^{|A|-1}(B^t)\}_{t \in \Delta}$ converges to $\mu^{|A|-1}(B^\infty)$ as $\mu$ is continuous. We also notice that $\{\max(\mu^{|A|-1}(B^t))\}_{t \in \Delta} = \{\max(B^{t+|A|-1})\}_{t \in \Delta}$ converges to the value $\max(\mu^{|A|-1}(B^\infty))$ and is a subsequence of the convergent sequence $\{\max(B^t)\}_{t \in \N}$, so both converge to the same value $\max(\mu^{|A|-1}(B^\infty)) = U$. So we conclude that $\max(\mu^{|A|-1}(B^\infty)) = U = \max(B^\infty)$ which is a contradiction with the inequality $\max(\mu^{|A|-1}(B^\infty)) < \max(B^\infty)$. Thus, $B^\infty$ must be consensus.

    Let $v$ be the value such that for all $i \in A$, $B_i^\infty = v$. As $\{\max(B^t)\}_{t \in \Delta}$ (which converges to $\max(B^\infty) = v$) is a subsequence of the convergent sequence $\{\max(B^t)\}_{t \in \N}$, both converge to the same value, so $\lim_{t\rightarrow \infty} \max(B^t) = v$. The same applies for the minimum, so $\lim_{t\rightarrow \infty} \max(B^t) = v = \lim_{t\rightarrow \infty} \min(B^t)$, and as $\max(B^t) \leq B_i^t \leq \min(B^t)$ for all $i \in A$ (Lem.\ref{lemboundsA}), by the Squeeze Theorem\cite{Sohrab:14}, $A$ converges to consensus. 

\end{proof}

\end{document}